\documentclass[pra,aps,twocolumn,notitlepage,superscriptaddress,showpacs,nofootinbib]{revtex4-2}

\usepackage{enumerate,appendix}
\usepackage{amsmath,amssymb,amsthm,amsfonts,mathrsfs}
\usepackage{mathtools}
\usepackage{bbold} 
\usepackage{optidef}

\usepackage{hyperref}
\usepackage{stmaryrd}
\usepackage{bbm}

\usepackage{graphicx,color,cases} 
\usepackage{xfrac}  
\definecolor{DarkGreen}{rgb}{0.1,0.5,0.1}
\definecolor{DarkRed}{rgb}{0.5,0.1,0.1}
\definecolor{DarkBlue}{rgb}{0.1,0.1,0.5}

\allowdisplaybreaks

\newtheorem{theorem}{Theorem} 
\newtheorem{lemma}[theorem]{Lemma} 
\newtheorem{corollary}[theorem]{Corollary}
 
\newtheorem{proposition}[theorem]{Proposition}
\theoremstyle{definition} 
\newtheorem{definition}[theorem]{Definition}

\newtheorem{remark}{Remark}
\numberwithin{equation}{section}
 
\def\>{\rangle} 
\def\<{\langle}
\def\Pr{{\rm Pr} }

\def\bu{{\bf u}} 
\def\bv{{\bf v}} 
\def\bx{{\bf x}} 
\def\by{{\bf y}} 
\def\bz{{\bf z}} 

\def\pu{p_{\bf u}}

\def\py{p_{\bf y}} 
\def\pz{p_{\bf z}} 

\def\psiu{\psi_{\bf u}}

\def\psiy{\psi_{\bf y}} 
\def\psiz{\psi_{\bf z}}

\def\QFI{{\rm QFI}}

\def\Pr{{\rm Pr} }
\DeclareMathOperator{\Var}{Var} 
\DeclareMathOperator{\var}{Var} 
\DeclareMathOperator{\wt}{wt}  
\DeclareMathOperator{\tr}{Tr}

\newcommand{\norm}[1]{\left\| #1 \right\|}	

\newcommand{\ket}[1]{\left\lvert #1 \right\rangle}
\newcommand{\bra}[1]{\left\langle #1 \right\rvert}
\NewDocumentCommand\ketbra{+m+g}{%
  \IfNoValueTF{#2}
    {\left\lvert #1 \right\rangle \left\langle #1 \right\vert}
  {\left\lvert #1 \right\rangle \left\langle #2 \right\rvert}%
}
\NewDocumentCommand\braket{+m+g}{%
  \IfNoValueTF{#2}
    {\left\langle #1 \vert #1 \right\rangle}
  {\left\langle #1 \vert #2 \right\rangle}%
}

\newif\ifcomment
\commenttrue
 
\begin{document} 

\title{Describing quantum metrology with erasure errors using weight distributions of classical codes}

\author{Yingkai Ouyang}
\email{oyingkai@gmail.com}
\affiliation{Department of Physics \& Astronomy, University of Sheffield, Sheffield, S3 7RH, United Kingdom}
\affiliation{Department of Electrical and Computer Engineering, National University of Singapore, Singapore 117583}
\affiliation{Centre of Quantum Technologies, National University of Singapore, Singapore 117543}

\author{Narayanan Rengaswamy}
\email{narayananr@arizona.edu}
\affiliation{Department of Electrical and Computer Engineering, University of Arizona, Tucson, AZ 85721, USA}


\thanks{Most of this work was conducted when NR was with the Department of Electrical and Computer Engineering, Duke University, Durham, North Carolina 27708, USA.}

\begin{abstract} 
Quantum sensors are expected to be a prominent use-case of quantum technologies, 
but in practice, noise easily degrades their performance. 
Quantum sensors can for instance be afflicted with erasure errors.
Here, we consider using quantum probe states with a structure that corresponds to classical $[n,k,d]$ binary block codes of minimum distance $d \geq t+1$.
We obtain bounds on the ultimate precision that these probe states can give for estimating the unknown magnitude of a classical field after at most $t$ qubits of the quantum probe state are erased.
We show that the quantum Fisher information is proportional to the variances of the weight distributions of the corresponding $2^t$ shortened codes. 
If the shortened codes of a fixed code with $d \geq t+1$ have a non-trivial weight distribution, then the probe states obtained by concatenating this code with repetition codes of increasing length enable asymptotically optimal field-sensing that passively tolerates up to $t$ erasure errors.
\end{abstract}

\maketitle

\section{Introduction} \label{sec:intro}

Quantum sensors which promise to estimate physical parameters with unprecedented precision have yet to realize their full potential because of decoherence. One approach to combat the effects of decoherence would be to use probe states chosen from quantum error correction codes, and perform active quantum error correction \cite{kessler2014quantum,dur2014improved,arrad2014increasing,unden2016quantum,matsuzaki2017magnetic,Zhou2018,layden2019ancilla,gorecki2019quantum}.
However, in lieu of active quantum error correction protocols, which remain challenging to implement, it is pertinent to understand the extent to which the advantage promised by quantum metrology can persist.
In this paper, we consider using probe states constructed from classical error correction codes, with no requirement of any quantum error correction protocol.
This approach to robust quantum metrology will be compatible with future protocols that are focused on fault-tolerant quantum metrology, since the probe states considered here are code states of quantum CSS codes~\cite{CSS97}. Hence, our strategy serves as a proposal both for near-term resource-limited schemes as well as long-term fault-tolerant architectures.

The research area of quantum metrology is very broad, and in this paper, we focus on the widely studied problem of field-sensing using quantum resources. 
We first describe the field-sensing scheme in the ideal noiseless setting. 
Here, the quantum resource is an $n$-qubit probe state $\rho$.
Mathematically, the action of a classical field that interacts linearly with an ensemble of $n$ qubits can be modelled using the unitary evolution 
\begin{align}
U_\theta = \exp(-i \theta H),
\end{align} 
where $\theta$ is an unknown phase proportional to the magnitude of the classical field strength, and $H$ is the Hamiltonian experienced by the qubits,
given by 
\begin{align}
H = Z_1 + \dots +Z_n,
\end{align} where $Z_j$ denotes the Pauli operator that applies $Z=|0\>\<0|-|1\>\<1|$ on qubit $j$ and the identity operator on all other qubits.
The unitary $U_\theta$ acts on the probe state $\rho$, which evolves into 
\begin{align}
\rho_\theta = U_\theta \rho U_\theta^\dagger.
\end{align} 
Then, an observable $M$ is measured on the state $\rho_\theta$ which depends on $\theta$.
This process is repeated for many preparations of the probe state $\rho$, and the measurement data is recorded.
Based on statistics of the measurement results, we can construct a classical estimator $\hat \theta$ that estimates $\theta$.
From $\hat \theta$, we can infer the magnitude of the classical field.
Typically in quantum metrology, we aim to find a locally unbiased estimator $\hat \theta$ that has the minimum variance $\Var(\hat \theta)$, because in general, the optimal observable $M$ will depend on the true value of $\theta$.

In a noiseless setting, field-sensing using quantum resources does offer a quantum advantage.
Using only classical resources (modelled by separable states), the optimal $\Var(\hat \theta)^{-1}$ scales linearly with the number of spins $n$.
Once we use a quantum probe state (that can have entanglement), there can be a quadratic scaling in the optimal $\Var(\hat \theta)^{-1}$. 
Namely, in a noiseless scenario, the optimal $\Var(\hat \theta)^{-1}$ is equal to $n^{2}$, and is achieved using an $n$-qubit GHZ state $(|0\>^{\otimes n} + |1\>^{\otimes n})/\sqrt 2$ as the probe state. However, relying on the GHZ state to achieve a quantum advantage in field-sensing is vulnerable to noise, because even a single erasure error or phase error renders the GHZ state completely classical.
This begs the question as to what type of probe states can be inherently robust for field-sensing. 

In this paper, we focus on what we call ``robust field-sensing''. 
In this model, noise only afflicts the initial probe state $\rho$ and all other aspects of the quantum sensing protocol remain noiseless. 
Robust field-sensing is a good approximation to the scenario when the majority of the noise occurs during preparation of the probe state, and when one is waiting for the signal to accumulate. 
Here, we consider erasure errors, where we know that qubits labeled by a set $E \subset \{1,\dots, n\}$ have been removed. 
The quantum channel that models erasure errors on $E$ is $\mathcal N_E$, which applies the partial trace on every qubit 
included in $E$. 
 
Erasure errors occur naturally in certain types of qubits, such as in dual rail qubits \cite{Knill2001_N}. A zero state and a one state in a dual rail qubit are represented by $|0_D\> = |0\>|1\>$ and $|1_D\> = |1\>|0\> $ respectively, where $|0\>$ represents the vacuum state and $|1\> $ represents a single photon state. With the loss of a single photon, the $|0_D\>$ state and $|1_D\>$ both can relax to $|0\>|0\>$. Determining whether each dual rail qubit is in the $|0\>|0\>$ state allows us to identify which dual rail qubits have been erased \cite{Knill2001_N}.

When a qubit is defined on the two lowest energy levels of a physical system, there can be leakage \cite{aliferis2007leakage} of the qubit to higher energy levels via energy excitations. The detection of such leakage errors allows us to pinpoint which qubits suffer from these leakage errors, and this can also be interpreted as an erasure error. Physical systems where this can occur include superconducting qubits \cite{werninghaus2021leakage} and atomic systems \cite{vala2005quantum}. Moreover, it has been recently shown how over 99\% of naturally occurring types of errors can also be converted into erasure errors in neutral atom qubits \cite{Wu2022}.

Usually, in coding theory, erasure errors are the first type of noise that one would consider \cite{sloane,nielsen-chuang}.
In view of this, we consider an erasure noise model that acts on quantum probe states in this paper.

GNU codes\cite{ouyang2014permutation}, 
an important subset of permutation-invariant quantum codes \cite{ouyang2015permutation,OUYANG201743,ouyang2019permutation,ouyang2021permutation},
have been studied in the context of being used for field-sensing in the presence of erasure errors \cite{ouyang2019robust}, but unfortunately GNU codes are only a very special family of codes.
In the field of quantum error correction, there are also many other important families of quantum error correction codes, such as CSS codes \cite{CSS97}, which have quantum states whose structure is based on the underlying classical codes.
Hence, the question arises as to how probe states with structure based on classical codes perform under the influence of erasure errors, in lieu of active quantum error correction.

In this paper, we address this gap; we calculate the performance of classical-code-inspired probe states in field-sensing, and show that they can be robust against some erasure errors.

\subsection{Field-sensing as a quantum state estimation problem}

In quantum state estimation \cite{hayashi2017book}, 
we have a known parametrized set of quantum states $\mathcal S = \{\rho_\theta : \theta \in \mathbb R\}$. Given multiple copies of an unknown state $\rho_\theta \in \mathcal S$, the task is to construct an unbiased estimator $\hat \theta$ to estimate $\theta$. 
The estimator is unbiased in the sense that $\mathbb E(\hat \theta) = \theta$.
We furthermore want to have the smallest possible mean squared error (MSE) ${\var(\hat \theta) = \mathbb E( (\hat \theta - \theta)^2)}$ \cite[Chapter 6.4]{hayashi2017book}.
This estimator $\hat \theta$ can be obtained from measuring an observable $M$ on $\rho_\theta$
The associated MSE arises from the error propagation formula and is given by \cite{ku1966notes}, \cite[Eqn. (1)]{sidhu2020geometric}:
\begin{align}
\var (\hat \theta)
    &= 
    \frac{\tr(\rho_{\theta} M^2)- \tr(\rho_{\theta} M)^2 }
    {\left|\frac{\partial}{\partial \theta}   \tr(\rho_{\theta} M) \right|^2},\label{eq:error-propagation-1-noiseless}
\end{align}
which holds in the limit where $\hat \theta \to \theta$.
The error propagation formula is a simple consequence of calculus and properties of the variance function. To see this, given the spectral decomposition $M = \sum_{i} m_i |m_i\>\<m_i|,$ note that by the Born rule, measurement of the state $\rho_\theta$ in the basis $\{|m_i\>: i\}$ collapses the state to $|m_i\>$ with probability $\<m_i| \rho_\theta| m_i\>$, and that the corresponding expected eigenvalue of $M$ is given by
\begin{align}
\mu = \tr(\rho_\theta M) = \sum_i m_i \<m_i| \rho_\theta|m_i\>.
\end{align}
Now, let $\hat \mu$ be an unbiased estimator of $\mu$, and assume that it is equal to $m_i$ with probability $\<m_i| \rho_{\theta} |m_i\>$.
Then it follows that $\mathbb E(\hat\mu) = \tr(\rho_\theta M)$ and \begin{align}
\var(\hat \mu) 
&= \sum_i m_i^2 \<m_i| \rho_{\theta} |m_i\>  
- \left( \sum_i m_i \<m_i| \rho_{\theta} |m_i\> \right)^2 
\notag\\
&= \tr(\rho_{\theta} M^2 ) - \tr(\rho_{\theta} M)^2.
\end{align}
In general, $\mu$ can be expressed as a function of $\theta$, where $\mu = f(\theta)$. 
Similarly, the estimator $\hat \mu$ can be written as a function of the estimator $\hat \theta$, where 
$\hat \mu = f(\hat \theta)$.

Assuming the continuity of $\<m_i| \rho_{\theta}|m_i\>$ as function of $\theta$, it follows that $f$ is continuous with respect to $\theta$. Hence for any $\theta_0$, we can write 
\begin{align}
    \mu = f_0 + f_1 (\theta -\theta_0) + O((\theta-\theta_0)^2).
\end{align}
It follows that 
\begin{align}
    \hat \mu = f_0 + f_1 (\hat \theta -\theta_0) + O((\hat \theta-\theta_0)^2),
\end{align}
and
\begin{align}
    \var(\hat \mu)
    =f_1^2 \var( \hat \theta )
    + O(\var(\hat \theta- \theta_0)^2).
\end{align}
From continuity of $\<m_i| \rho_{\theta}|m_i\>$ with respect to $\theta$, it follows that for all $\theta_0 \in \mathbb R$, we have
\begin{align}
f_1 =
\sum_i m_i
\<m_i| \left.\frac{\partial\rho_{\theta}}{\partial\theta}\right|_{\theta = \theta_0}|m_i\>
= \tr\left( \left. \frac{\partial  \rho_\theta}{\partial \theta}\right|_{\theta = \theta_0} M \right) 
.\label{eq:continuity-of-probability}
\end{align}
The error propagation formula then follows from substituting $\theta_0 = \theta$, which gives
Using \eqref{eq:continuity-of-probability}, it follows that 
 \begin{align}
\Var (\hat \theta  )  
= \Var(\hat \mu)  \tr(M\frac{\partial\rho_{\theta}}{\partial\theta} )^{-2}
+  O(\Var((\hat \theta -\theta)^2)). \label{eq:error-propagation-with-bigO}
\end{align}
This proves that the error propagation formula \eqref{eq:error-propagation-1-noiseless} holds in the limit where $\hat \theta \to \theta$. 

When $\hat \theta$ is a locally unbiased estimator of a fixed $\theta$, 
its minimum MSE can be calculated by solving the following optimization program:
\begin{mini}
{M {\ \rm Hermitian}}{\var(\hat \theta)} {}{}
\addConstraint{ \tr ( \rho_\theta M ) }{ =  \theta}{}.\label{mini:MSE}
\end{mini} 
The estimator $\hat \theta$ is a (globally) unbiased estimator if $\hat \theta$ is locally unbiased for all values of $\theta$, which means that the optimal solution $M$ to this optimization problem is independent of $\theta$. In general, it is not possible to find an optimal $M$ that is independent of $\theta$, and hence it is typical to focus on the scenario where $\hat \theta$ is a locally unbiased estimator.

Now, substituting \eqref{eq:error-propagation-1-noiseless} and rewriting the optimization problem, we get equivalently
\begin{mini}
{M {\ \rm Hermitian}}{    \tr(\rho_{\theta} M^2)- \tr(\rho_{\theta} M)^2 } {}{}
\addConstraint{ \tr ( \rho_\theta M ) }{ =  \theta}{}
\addConstraint{ \tr \left( \frac{\partial \rho_\theta}{\partial \theta} M \right) }{ = 1}{},\label{mini:MSE-nofrac}
\end{mini} 
which has the advantage of being written explicitly as a convex optimization program, with a convex quadratic objective function and linear constraints.
Hence, this program models the quantum state estimation problem. 

To establish the connection between quantum state estimation and field-sensing, we consider a specific family of states $\mathcal S$, namely $n$-qubit states such that for every $\theta \in \mathbb R$, we can write $\rho_\theta = U_\theta \rho U_\theta^\dagger$, where $U_\theta = e^{-i \theta H}$ and $H = Z_1 + \dots + Z_n$.
Since we are locally estimating $\hat \theta$, without loss of generality, we can take $\theta$ to be in the neighborhood of zero.

\subsection{The quantum Cram\'er-Rao bound}

In quantum metrology, the Hermitian operator $L_\theta$, known as the symmetric logarithm derivative (SLD), is defined implicitly via the equation  
\begin{align}
\frac{\partial \rho_\theta}{\partial \theta} = \frac{1}{2}\left(L_\theta \rho_\theta + \rho_\theta L_\theta \right). \label{sld}
\end{align} 
For example, when $\rho$ is a pure state and 
$U_\theta = e^{-i \theta H}$, the corresponding SLD can be written as 
\begin{align}
L_\theta = 2i(\rho_\theta H-H \rho_\theta).\label{eq:perfect-SLD}
\end{align}
While the solution $L_\theta$ to \eqref{sld} is not necessarily unique, the quantity $\tr(\rho_\theta L_\theta^2)$ is well-defined \cite{liu2016quantum}.
This is because $L_\theta$ admits a basis-independent representation as a solution of the Lyapunov differential equation \cite[Eq. (94)]{sidhu2020geometric}.
Because of this, we can define the quantity
\begin{align}
Q(\rho_\theta, \frac{\partial\rho_\theta}{\partial\theta}) \coloneqq \tr (\rho_\theta L_\theta^2 )
\label{eq:qfidef}
\end{align}
based on the SLD, and this quantity depends only on $\rho_\theta$ and $\frac{\partial\rho_\theta}{\partial\theta}$.

The quantity $Q(\rho_\theta, \frac{\partial\rho_\theta}{\partial\theta})$ is equal to what is known as the quantum Fisher information (QFI), which plays a central role in quantum metrology. The QFI can also be interpeted as a metric \cite{wootters1981statistical_QFI,braunstein1994statistical_QFI,fujiwara1995quantum}, which in this case depends on the probe state and the generator $H$. However, evaluating the QFI explicitly requires the full spectral decomposition of the probe state, and may thus be challenging to find in general.
While the QFI can be defined in many different but equivalent ways, we use \eqref{eq:qfidef} for its structural simplicity.  
 
The QFI's importance arises from its role in the celebrated quantum Cram\'er-Rao bound (QCRB), proven by Helstrom and Holevo~\cite{Helstrom68,Helstrom1976,Holevo2011}. 
Namely, the QCRB states that the minimum variance of $\hat \theta$ must satisfy the following inequality  
\begin{align}
\min\{\Var(\hat \theta) : M=M^\dagger, \tr (\rho_\theta M) = \theta\}  \ge Q(\rho_\theta, \frac{\partial\rho_\theta}{\partial\theta})^{-1}. \label{variance-qfi}
\end{align}
For the case of estimating a single parameter, which is the focus of this paper, the QCRB can be tight, which means that the QFI can tell us what the ultimate precision is for quantum metrology. In particular, equality in the QCRB is attained by using a $\theta$-dependent observable $M = \theta {\bf 1} + L_{\theta} / Q(\rho_\theta, \frac{\partial\rho_\theta}{\partial\theta})$\cite[Eq (108)]{sidhu2020geometric}.
For notational simplicity, when it is clear what 
$\rho_\theta$ and $\frac{\partial\rho_\theta}{\partial\theta}$ are from the context, 
we will write $\QFI = Q(\rho_\theta, \frac{\partial\rho_\theta}{\partial\theta})$.

Furthermore, when $\rho_\theta = U_\theta \rho U_\theta^\dagger$, it is easy to see that the $M$ that saturates the QCRB is also locally unbiased and hence a feasible solution to the optimisation problem \eqref{mini:MSE}. 
To see this, note that $\tr(\frac{\partial \rho_\theta}{d\theta}) = 0$, since every $\rho \in \mathcal S$ must have unit trace.  
Also, from the SLD equation we have $\frac{\partial \rho_\theta}{d\theta} = \frac{1}{2}(L_\theta \rho_\theta + \rho_\theta L_\theta)$. 
It follows that 
$\tr(\rho_\theta L_\theta) = \frac{1}{2} \tr(\rho_\theta L_\theta) + \frac{1}{2} \tr(L_\theta\rho_\theta ) =  \tr(\frac{\partial \rho_\theta}{d\theta} )  = 0$.  
This implies that 
\begin{align}
\hat \theta &= \tr(\rho_\theta M) = \tr\left( \rho_\theta \left( \theta {\bf 1} + L_{\theta} / Q\left( \rho_\theta, \frac{\partial\rho_\theta}{\partial\theta} \right) \right) \right) 
\notag\\
&= 
\tr(\rho_\theta  \theta {\bf 1} )
= \theta,
\end{align}
which proves that for the phase estimation problem that we consider, the optimal observable gives rise to a locally unbiased estimator $\hat \theta$ with variance that saturates the QCRB.

Since the QFI can be complicated to analyze analytically in general, we can appeal to some bounds on the QFI. 
A lower bound on the QFI that arises from its relationship with the Bures distance \cite[Lemma 6]{PRX-random-states} is given by
\begin{align}
Q(\rho_\theta, \frac{\partial\rho_\theta}{\partial\theta})
\ge \| [H, \rho] \|_1^2.
\end{align}
From the fact that the 1-norm is lower bounded by the 2-norm, we get
\begin{align}
     Q(\rho_\theta, \frac{\partial\rho_\theta}{\partial\theta})\ge \| [\rho ,{ H}] \|_2^2 
     &=
      2 \tr (\rho^2 {  H}^2) -  2 \tr( \rho {  H} \rho {  H} ) \label{eq:QFI bound}.
\end{align} 
This bound is tight when $\rho $ is a pure state, which corresponds to the case where zero erasures occur. 
The QFI is also at most the variance of the observable $H$ (see \cite[Eq (98)]{sidhu2020geometric}): \begin{align}
    Q(\rho_\theta, \frac{\partial\rho_\theta}{\partial\theta}) \le 
   4 \tr(\rho  {  H}^2) - 4 \tr(\rho  {  H})^2.
   \label{eq:qfi-upperbound-variance}
\end{align}
This bound is tight when $\rho$ is a pure state.
Using \eqref{eq:qfi-upperbound-variance}, when $U_\theta = \exp(-i \theta H)$, the optimal entangled $\rho$ corresponds to the GHZ state $(|0\>^{\otimes n} + |1\>^{\otimes n})/\sqrt{2}$ and has a QFI equal to $n^2$, which achieves the so-called Heisenberg scaling. In comparison, the optimal QFI for separable states is achieved by $|+\>^{\otimes n}$, and has a QFI equal to $n$, which is the standard quantum limit (SQL).

Note that when even a single erasure occurs on the GHZ state, it becomes separable and loses its quantum advantage. Therein lies the question of what probe states have a QFI that is robust against erasure errors.
In this paper, we will construct probe states for robust field-sensing from classical error-correcting codes, and evaluate their corresponding QFI lower and upper bounds under a noise model that introduces a finite set of erasures. 
We emphasize that our protocol does not involve quantum error correction procedures. 

\subsection{Classical codes and their corresponding probe states}

The classical codes that we consider are sets of length $n$ binary strings, called codewords, so these sets are subsets of $\{0,1\}^n$. Given any length $n$ binary classical code $C \subset \{0,1\}^n$, we let $|\psi_C\>$ denote a pure state of the form
\begin{align}
    |\psi_C\> \coloneqq \frac{1}{\sqrt{|C|}} \sum_{\bx \in C} |\bx\> \label{eq:state-defi},
\end{align}
where for ${\bf x}= (x_1,\dots ,x_n) \in \{0,1\}^n$ we define $|\bx\> \coloneqq |x_1\> \otimes \dots \otimes |x_n\>$.
We propose to use $|\psi_C\>$ as the ideal probe state for robust field-sensing. 
By choosing appropriate codes, one might leverage fault-tolerant state preparation techniques in the quantum error correction literature to prepare these fixed states, but we do not discuss this aspect further in this paper.

A classical binary code $C$ is said to be linear if it is a group under the binary addition operator, and is said to be self-orthogonal if $\bx \cdot \by = 0$ for all $\bx,\by\in C$, where $\cdot$ denotes the inner product over the binary field.
An $[n,k,d]$ binary linear code $C$ encodes $k$ information bits into $n$ code bits, and the minimum number of $1$s in any codeword of $C$ is $d$, called its minimum distance.
One important family of quantum codes are CSS codes~\cite{CSS97},
which can be constructed from a pair of classical codes $C_1, C_2$ that satisfy $C_2 \subseteq C_1$. As a special case, a CSS code can be constructed from any binary linear code $C$ that is also self-orthogonal, by setting $C_2 = C$ and $C_1 = C^\perp$, the dual code of $C$.
A short background on classical codes and quantum CSS codes is given in Appendix~\ref{sec:css_codes}. 

When interpreted through the lens of quantum error correction (QEC), if $C$ is taken to be a linear self-orthogonal code, then $|\psi_C\>$ corresponds to the logical $\ket{++\cdots +}$ state of a CSS code by appropriately identifying the CSS code's logical $X$ operators (including the $X$-type stabilizers) from the code $C$.
Similarly, if $C$ is taken to provide only the $X$-type stabilizers, then the above is just the logical $\ket{00\cdots 0}$ 
state~\cite[Section 10.4.2]{nielsen-chuang}.
As a trivial case, the state $|\psi_C\>$ is also the unique code state (up to a global phase) of the CSS code defined by $C_1 = C_2 = C$.

\begin{figure}
  \centering
  \includegraphics[width=0.6\linewidth]{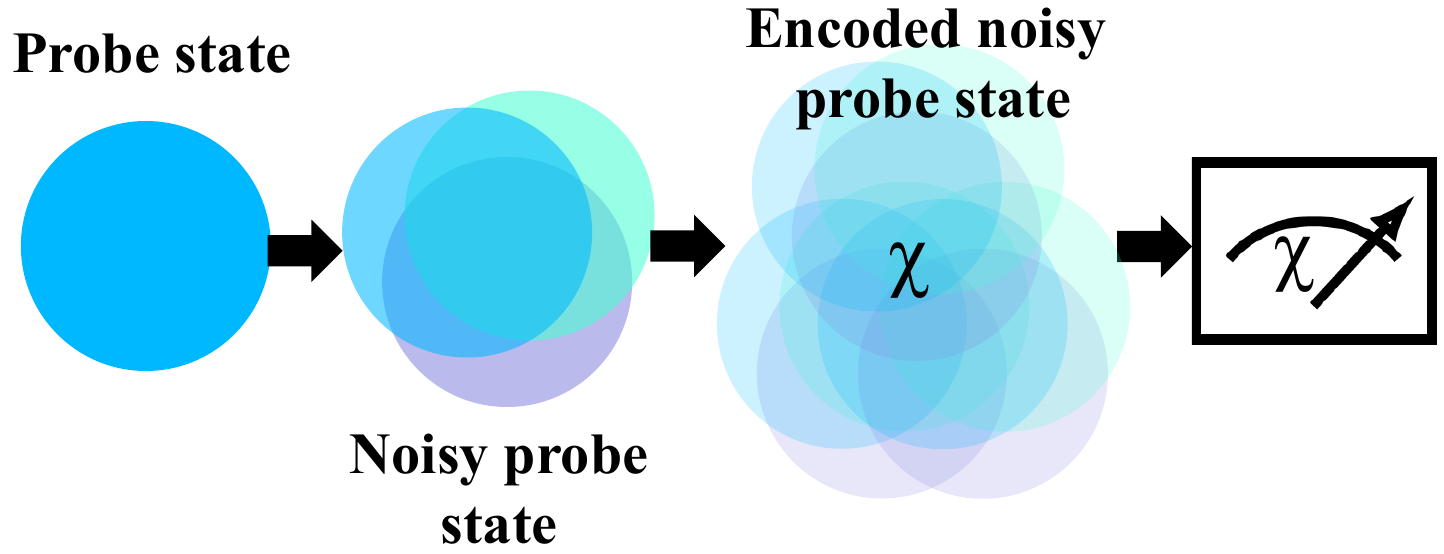}
  \caption{A cartoon sketch of the robust metrology problem. First a probe state is prepared. Then noise occurs on the probe state. Subsequently, the probe state picks up information $\theta$. An estimator $\hat \theta$ of the true parameter $\theta$ is obtained after measuring an appropriate physical observable on the quantum state.} 
  \label{fig:cartoon}
\end{figure}

\subsection{The main result and its implications}  
In this paper, given any classical code $C$ with a minimum distance that is at least $t+1$, we obtain upper and lower bounds on the QFI after any $t$ qubits are erased from the quantum probe state $\rho_C =|\psi_C\>\<\psi_C|$.
When erasure errors occur on a subset $E= \{j_1,\dots, j_t\}$ of qubits in the probe state $\rho_C$, the resultant state is just the partial trace of $\rho_C$ over the qubits labeled by $E$.
We denote this state as $\rho_C[E] = \mathcal N_E(\rho_C)$, where $\mathcal N_E$ denotes an erasure channel on the qubits labelled by the set $E$, with $|E| \le t$.
We say that $t$ erasures occur if $|E|=t$. If $E$ comprises of $t'$ consecutive qubits, we say that $t'$ burst erasures have occurred. 
In the robust field-sensing problem, to bound the MSE of the estimator $\hat \theta$, it suffices to obtain bounds on the QFI
\begin{align}
Q_E(C) \coloneqq   Q\left( U_\theta \rho_C[E] U_\theta^\dagger, \frac{\partial }{\partial\theta}(U_\theta \rho_C[E]U_\theta^\dagger) \right),
\end{align}
where $U_\theta = \exp(-i \theta(Z_1+\dots + Z_{n-|E|}))$.

In this scenario, we find that using our code-inspired probe states, after at most $t$ erasure errors have occurred, $Q_E(C)$ scales with the variances of the weight distributions of $2^t$ shortened codes of $C$. 
Our result applies in a very general setting; aside from a distance criterion we impose on the classical code $C$, we make no other assumptions. Hence, $C$ can in general be a non-linear binary code.
Moreover, we show that if $C$ is any constant-length code concatenated with inner repetition codes of length linear in $n$, then $Q_E(C)$ is at least quadratic in the code length $n$ and the concatenated code can also withstand a linear number of burst erasures in $n$ (see Corollary~\ref{cor:QFI_scaling}).

An implication of our result is that, given any quantum CSS code of constant length, we can concatenate it with repetition codes of length linear in $n$.
If we do so and pick an appropriate state in the CSS codespace, the QFI under erasure errors is boosted by concatenation with the inner repetition codes. The operational significance is that (a) CSS codes are well-understood in quantum coding theory; (b) we can achieve robust field-sensing with CSS states, without any error correction performed; and (c) our framework for robust field-sensing is compatible with subsequent protocols where quantum error correction on CSS codes is required \cite{kessler2014quantum,dur2014improved,arrad2014increasing,unden2016quantum,matsuzaki2017magnetic,Zhou2018,layden2019ancilla,gorecki2019quantum}.

Now let us outline the structure of the paper.
In Section \ref{sec:2-main-result}, we give the main results of our paper, which are bounds on the precision of estimating $\theta$ after $t$ erasures have occurred on a probe state $|\psi_C\>$ constructed from a classical code $C$. 
In Section \ref{sec:2b-RM-codes}, we consider the example where $C$ is a binary Reed-Muller code with parameters RM(1,3), and evaluate upper and lower bounds on its QFI in the noiseless case as well as when there is a single erasure.
In Section \ref{sec:2c-erasures}, we introduce notation related to having multiple erasures, and define a disjointness property for partitions of $C$ that we need to establish our main results.
In Section \ref{sec:2d-lower-bounds}, we give a lower bound on the QFI of a probe state $|\psi_C\>$ after $t$ erasure errors have occurred. This lower bound is related to the variances of the weight distributions of  
$2^t$ shortened codes of $C$.
We also show how concatenating $C$ with inner repetition codes can allow the QFI to scale quadratically with $n$ as long as the length of the outer code is held constant and the number of erasures $t$ remains bounded by the disjointness criterion. Our results imply that our probe states can also tolerate a linear number of burst erasures.  
In Section \ref{sec:2e-upper-bounds}, we give corresponding upper bounds on the QFI.
In Section \ref{sec:3-measurements}, we use the error propagator formula \eqref{eq:error-propagation-with-bigO} to evaluate the MSE for an observable motivated from the SLDs of pure states, and show that the MSE exhibits the same behaviour as the inverse of the QFI. 
In Section \ref{sec:4-examples}, we explain how our results can work with explicit codes.
Finally in Section \ref{sec:5-discussions}, we summarize our results and discuss what we think are interesting problems to consider in the future.

\section{Bounds on the QFI after erasures} \label{sec:2-main-result}
When $t$ erasure errors occur, we can always write down the labels of the erased qubits to be $j_1,\dots, j_t$ and let $E=\{j_1,\dots, j_t\}$ denote the corresponding set of erasures. 
Without loss of generality, we can assume that 
$j_1<\dots<j_t$. 
For notational simplicity, we will denote $\rho = \rho_C[E]$ in the rest of the paper whenever the set $E$ and the code $C$ is clear from the context, and set $H = Z_1+\dots + Z_{n-t}$.

\subsection{Example: A probe state from a [[8,3,2]] Reed-Muller code}
\label{sec:2b-RM-codes}
 
Consider the $[[ 8,3,2 ]]$ quantum Reed-Muller (QRM) code~\cite{Campbell-blog16,Rengaswamy-jsait20} described by two classical binary codes $C_2$ and $C_1$ defined as $C_2 \coloneqq \text{RM}(0,3)$ and $C_1 \coloneqq \text{RM}(1,3)$, respectively.
Due to the properties of RM codes \cite{Macwilliams-1977}, we have $C_2 \subset C_1$ and the dimensions are $\text{dim}(C_2) = 1$ and $\text{dim}(C_1) = 4$.
The standard generator matrix for $C_1$ is given by
\begin{align}
G(C_1) = 
\begin{bmatrix}
1 & 1 & 1 & 1 & 1 & 1 & 1 & 1 \\
\hline
0 & 1 & 0 & 1 & 0 & 1 & 0 & 1 \\
0 & 0 & 1 & 1 & 0 & 0 & 1 & 1 \\
0 & 0 & 0 & 0 & 1 & 1 & 1 & 1
\end{bmatrix} = 
\begin{bmatrix}
G(C_2) \\
\hline
G(C_1/C_2)
\end{bmatrix}.
\end{align}
The $X$-stabilizers are given by $C_2$, the length $8$ repetition code, and the $Z$-stabilizers are given by $C_1^{\perp} = C_1$ since $C_1$ is the $[8,4,4]$ extended Hamming code that is self-dual.
The canonical generators for the logical $X$ operators correspond to degree-$1$ monomials that generate the space $C_1/C_2$, namely $\bar{X}_1 = X_2 X_4 X_6 X_8, \bar{X}_2 = X_3 X_4 X_7 X_8, \bar{X}_3 = X_5 X_6 X_7 X_8$.
Therefore, for $x_1,x_2,x_3 \in \{0,1\}$, the logical computational basis states can be written as
\begin{align}
\ket{x_1 x_2 x_3}_L \equiv \frac{1}{\sqrt{2}} \bar{X}_1^{x_1} \bar{X}_2^{x_2} \bar{X}_3^{x_3} \left( \ket{00000000} + \ket{11111111} \right).
\end{align}
We choose the probe state for metrology as
\begin{align}
\ket{\psi} = \ket{+++}_L \equiv \frac{1}{4} \sum_{c \in C_1} \ket{c}.
\end{align}
First, let us assume that the channel erases the first qubit, so that the resulting reduced density matrix is $\rho = \text{Tr}_1\left[ \ketbra{\psi} \right]$. 
The generator matrix, $G_1^p = G(C_1^p) \in \{0,1\}^{4 \times 7}$, for the code $C_1$ punctured in the first position, which is the standard $[7,4,3]$ Hamming code, is the matrix $G(C_1)$ with the first column removed. 
The last $3$ rows of $G_1^p$, denoted as the matrix $G_1^s = G(C_1^s) \in \{0,1\}^{3 \times 7}$, is a generator matrix for the dual of the Hamming code, which is the shortened $\text{RM}(1,3)$ code, also called the $[7,3,4]$ simplex code.
Therefore, the aforesaid reduced density matrix is
\begin{align}
\rho = \frac{1}{16} \sum_{c_1,c_2 \in C_1^s} \ketbra{c_1}{c_2} + \frac{1}{16} \sum_{c_1,c_2 \in C_1^s} \ketbra{\underline{1} \oplus c_1}{\underline{1} \oplus c_2},
\end{align}
where $\underline{1}$ denotes the length $7$ vector with all entries equal to $1$.
In this example, $H = \sum_{i=1}^{7} Z_i$. To obtain a lower bound on the QFI of the probe state after the first qubit is erased, it suffices to
calculate $2 \tr\left( \rho^2 H^2 \right) - 2 \tr(\rho H \rho H) = \norm{[\rho, H]}_2^2$.
We first observe that
\begin{widetext}
\begin{align}
H^2 & = 7 I_{128} + 2 \sum_{\substack{i,j=1\\i < j}}^{7} Z_i Z_j, \\
\rho^2 & = \frac{1}{256} \sum_{c_1, c_2, c_1', c_2' \in C_1^s} \bigg[ \ket{c_1} \braket{c_2}{c_1'} \bra{c_2'} + \ket{\underline{1} \oplus c_1} \braket{\underline{1} \oplus c_2}{\underline{1} \oplus c_1'} \bra{\underline{1} \oplus c_2'} \bigg] \\
  & = \frac{8}{256} \sum_{c_1,c_2' \in C_1^s} \left[ \ketbra{c_1}{c_2'} + \ketbra{\underline{1} \oplus c_1}{\underline{1} \oplus c_2'} \right] \\
  & = \frac{1}{32} \sum_{c_1,c_2 \in C_1^s} \bigg[ \ketbra{c_1}{c_2} + \ketbra{\underline{1} \oplus c_1}{\underline{1} \oplus c_2} \bigg].
\end{align}
\end{widetext}
Then we can calculate  
 \begin{widetext}
\begin{align}
\rho^2 H^2 & = \underset{A}{\underbrace{\frac{7}{32} \sum_{c_1, c_2 \in C_1^s} \bigg[ \ketbra{c_1}{c_2} + \ketbra{\underline{1} \oplus c_1}{\underline{1} \oplus c_2} \bigg]}} + \underset{B}{\underbrace{\frac{2}{32} \sum_{c_1, c_2 \in C_1^s} \sum_{\substack{i,j=1\\i<j}}^{7} \bigg[ \ketbra{c_1}{c_2} Z_i Z_j + \ketbra{\underline{1} \oplus c_1}{\underline{1} \oplus c_2} Z_i Z_j \bigg]}}, \\ 
\tr(A) & = \frac{7}{32} \times (2 \times 8) = \frac{7}{2}, \\
\tr(B) & \overset{\text{(a)}}{=} \frac{1}{16} \sum_{c_1 \in C_1^s} \sum_{\substack{i,j=1\\i<j}}^{7} 2 (-1)^{c_1 (e_i + e_j)^T} \\
  & \overset{\text{(b)}}{=} \frac{1}{8} \bigg[ \binom{7}{2} + \bigg\{ \binom{4}{2} + \binom{3}{2} - \left( 21 - \binom{4}{2} - \binom{3}{2} \right) \bigg\} \times 7 \bigg] \\
  & = 0 \\
\Rightarrow 2 \tr(\rho^2 H^2) & = 7.
\end{align}
 \end{widetext}
In step (a) the vectors $e_i$ and $e_j$ denote the standard basis vectors for $\{0,1\}^7$ with a single entry $1$ in the $i$-th and $j$-th entry, respectively, and zeros elsewhere.
We have used the fact that $Z_i Z_j \ket{c_2} = (-1)^{c_2 (e_i + e_j)^T} \ket{c_2}$ and $c_1 = c_2$ for the trace to be non-zero.
Furthermore, in step (b) we have used the fact that all non-zero codewords of the simplex code $C_1^s$ have weight exactly $4$.
Next we calculate
 \begin{widetext}
\begin{align}
\rho H & = \frac{1}{16} \sum_{c_1, c_2 \in C_1^s} \sum_{i=1}^{7} \bigg[ \ketbra{c_1}{c_2} Z_i + \ketbra{\underline{1} \oplus c_1}{\underline{1} \oplus c_2} Z_i \bigg] \\
  & = \frac{1}{16} \sum_{c_1, c_2 \in C_1^s} \left( \sum_{i=1}^{7} (-1)^{c_{2,i}} \right) \bigg[ \ketbra{c_1}{c_2} - \ketbra{\underline{1} \oplus c_1}{\underline{1} \oplus c_2} \bigg] .
 \end{align}
 This implies that 
 \begin{align}
 \rho H \rho H & = \frac{1}{256} \bigg[ \sum_{c_1, c_2 \in C_1^s} \sum_{i=1}^{7} (-1)^{c_{2,i}} \left( \ketbra{c_1}{c_2} - \ketbra{\underline{1} \oplus c_1}{\underline{1} \oplus c_2} \right) \bigg] \bigg[ \sum_{c_1', c_2' \in C_1^s} \sum_{j=1}^{7} (-1)^{c_{2,j}'} \left( \ketbra{c_1'}{c_2'} - \ketbra{\underline{1} \oplus c_1'}{\underline{1} \oplus c_2'} \right) \bigg] \\
  & = \frac{1}{256} \sum_{c_1, c_2, c_2' \in C_1^s} \left( \sum_{i,j=1}^{7} (-1)^{c_{2,i} + c_{2,j}'} \right) \left( \ketbra{c_1}{c_2'} + \ketbra{\underline{1} \oplus c_1}{\underline{1} \oplus c_2'} \right) \\
  & = \frac{1}{256} \bigg[ \sum_{c_2 \in C_1^s} \sum_{i=1}^{7} (-1)^{c_{2,i}} \bigg] \bigg[ \sum_{c_1, c_2' \in C_1^s} \left( \sum_{j=1}^{7} (-1)^{c_{2,j}'} \right) \left( \ketbra{c_1}{c_2'} + \ketbra{\underline{1} \oplus c_1}{\underline{1} \oplus c_2'} \right) \bigg] \\
  & = 0 .
  \end{align}
Hence, it follows that $2 \tr(\rho H \rho H) = 0$, from which it follows that 
\begin{align}
\norm{[\rho, H]}_2^2 & = 2 \tr(\rho^2 H^2) - 2 \tr(\rho H \rho H) = 7  .
\end{align}
 \end{widetext}
If any qubit other than the first is erased, it can be easily verified that the resulting shortened code $C_1^s$, where the punctured bit takes the value $0$ in all codewords, has an identical weight distribution as for the above case of the first bit being erased.
A similar statement is true for the coset of this shortened code generated by adding the all $1$s vector to all codewords, corresponding to the second summation in $\rho$ above. 
Hence, if exactly one qubit out of the $8$ are erased, 
then the QFI is at least 7. 

Let us now calculate the QFI lower bound for the state $\ket{\psi}$ when there are no erasures.
First, we have 
\begin{align}
\rho_C & = \ketbra{\psi} = \frac{1}{16} \sum_{c_1,c_2 \in C_1} \ketbra{c_1}{c_2} = \rho^2.
\end{align} 
In this case, we can take  
$H = \sum_{i=1}^{8} Z_i$.
Then $H^2 = 8 I_{256} + 2 \sum_{i < j}^{8} Z_i Z_j$.
Hence,
\begin{align}
&\rho^2 H^2 \notag\\
 =& \frac{8}{16} \sum_{c_1,c_2 \in C_1}\!\!\!\!\! \ketbra{c_1}{c_2} + \frac{2}{16} \sum_{c_1,c_2 \in C_1} \sum_{\substack{i,j=1\\i < j}}^{8} \ketbra{c_1}{c_2} Z_i Z_j \\
  =& \frac{1}{2} \sum_{c_1,c_2 \in C_1}\!\!\!\!\! \ketbra{c_1}{c_2} + \frac{1}{8} \sum_{c_1,c_2 \in C_1}\sum_{\substack{i,j=1\\i < j}}^{8} (-1)^{c_2 (e_i + e_j)^T} \ketbra{c_1}{c_2}.
\end{align}
It is clear that the trace of the first term is $16/2 = 8$.
For the trace of the second term, we calculate
\begin{align}
&\frac{1}{8} \sum_{c_1 \in C_1} \sum_{\substack{i,j=1\\i < j}}^{8} (-1)^{c_1 (e_i + e_j)^T} \notag\\
   \overset{\text{(a)}}{=}& \frac{1}{8} \bigg[ 2 \binom{8}{2} + 14  \bigg\{ \binom{4}{2} + \binom{4}{2} - \left( \binom{8}{2} - 2 \binom{4}{2} \right) \bigg\} \bigg] \\
   =& 7 - 7    = 0.
\end{align}
Here, in step (a) we have used the fact that except the all-zeros codeword and the all-ones codeword, all codewords in $C_1$ have weight exactly $4$.
Therefore, we have $2 \tr(\rho^2 H^2) = 16$.
Next, we observe that
\begin{align}
\rho H & = \frac{1}{16} \sum_{c_1,c_2 \in C_1} \sum_{i=1}^{8} (-1)^{c_{2,i}} \ketbra{c_1}{c_2}.
\end{align}
So, we can calculate
\begin{align}
&  \rho H \rho H  
  \notag\\
   = &\frac{1}{256} \sum_{c_1,c_2,c_1',c_2' \in C_1} \sum_{i,j=1}^{8} (-1)^{c_{2,i} + c_{2,j}'} \ket{c_1} \braket{c_2}{c_1'} \bra{c_2'} \\
   =& \frac{1}{256} \bigg[ \sum_{c_2 \in C_1} \sum_{i=1}^{8} (-1)^{c_{2,i}} \bigg] \sum_{c_1,c_2' \in C_1} \sum_{j=1}^{8} (-1)^{c_{2,j}'} \ketbra{c_1}{c_2'} \\
   =& 0,
\end{align}
again because the only codeword weights in $C_1$ are $0,4,8$.
This implies $2 \tr(\rho H \rho H) = 0$ and thus 
\begin{align}
\norm{[\ketbra{\psi}, H]}_2^2 = 2 \tr(\rho^2 H^2) - 2 \tr(\rho H \rho H) = 16.
\end{align}
This means that the QFI of the Reed-Muller pure probe state is 16.
In contrast, if we use the optimal separable state, 
we have a QFI contribution of $n=8$.
Since this generator bound is tight for pure states, we observe that the introduction of a single erasure has more than halved the QFI lower bound for this Reed-Muller probe state.
We will show later that this situation can be improved by concatenating the chosen code $C = C_1$ with an inner repetition code.
Namely, the QFI lower bound becomes $7r^2$ when we use inner repetition codes of length $r$. 
For example, concatenating with repetition codes of lengths $3, 5$, and $8$ produces probe states consisting of $24, 40$, and $64$ qubits, with QFI lower bounds of $63, 175$, and $448$ respectively, assuming a single qubit is erased.
Hence, these QRM probe states outperform the optimal separable states at large sizes.

\subsection{Multiple erasures on a general probe state}
\label{sec:2c-erasures}
 We will now show that 
the erasure corrupted probe state $\rho_C[E]$ has a particularly simple form. 
But first, we have to introduce some notation that corresponds to the partition of the code $C$ into $2^t$ sets, each labeled by $C_{{\bf z},E}$, where ${\bf z}=(z_1,\dots,z_t)$ denotes a length $t$ binary string.
The set $C_{{\bf z},E}$ consists of all codewords ${\bf c}$ in $C$ that satisfy $c_{j_i} = z_i$ for all $i \in \{1,2,\ldots,t\}$, where $E = \{ j_1, j_2, \ldots, j_t \} \subsetneq \{ 1,2,\ldots,n \}$ as defined earlier.
Now, given a vector $\bx$, let $\bx[E]$ denote the vector obtained from $\bx$ after deleting (puncturing) all components labeled by the set $E$.
Then for all ${\bf z} \in \{0,1\}^t$, we define the length $n-t$ shortened codes of $C$,
\begin{align}
C_{{\bf z}}[E] \coloneqq \{\bx[E]: \bx \in C_{{\bf z},E} \}.
\end{align}
Note that $C_{{\bf z}}[E]$ is a non-linear code except when ${\bf z}$ is the all-zeros vector.
Also, if $C_{\bf z}[E]$ is not an empty set
\begin{align}
       p_{\bf z}  & \coloneqq \frac{|C_{{\bf z}}[E]|}{|C|} = \frac{|C_{{\bf z}, E}|}{|C|},\\
\ket{\psi_{\bf z}} & \coloneqq \frac{1}{\sqrt{| C_{\bf z}[E] |}} \sum_{{\bf x} \in C_{\bf z}[E]} \ket{{\bf x}}.
\end{align}
If $C_{\bf z}[E]$ is an empty set, we let 
$p_{\bf z} =0$ and $\ket{\psi_{\bf z}}=0$.
Using this notation, we will now obtain a simple expression for $\rho_C[E]$, the probe state after qubits in $E$ are erased.  

\begin{proposition}
\label{prop:erasure-t-qubits}
Let $C$ be a binary code of length $n$, and let $E = \{ j_1, j_2, \ldots, j_t \} \subsetneq \{ 1,2,\ldots,n \}$.  
Let $\rho_C = |\psi_C\>\<\psi_C|$, where $|\psi_C\>$ is as given in \eqref{eq:state-defi}. 
When the qubits belonging to $E$ are erased from $\rho_C$, the candidate probe state becomes 
\begin{align}
\rho_{C}[E] = \frac{1}{|C|}
\sum_{{\bf z} \in \{0,1\}^t }
\sum_{\bx,\by \in C_{\bf z}[E]} |\bx\>\<\by| 
=
\sum_{{\bf z} \in \{0,1\}^t }
p_{\bf z} |\psi_{\bf z}\>\<\psi_{\bf z}|.
\end{align} 
\end{proposition} 

In general, the codes $C_{\bz}[E]$ for distinct values of $\bz \in \{0,1\}^t$ need not be disjoint, and the states $|\psi_\by\>$ and $|\psi_\bz\>$ need not be pairwise orthogonal.
Proposition \ref{prop:erasure-t-qubits} only gives a spectral decomposition of $\rho_C[E]$ when the codes $C_{\bz}[E]$ are disjoint codes for distinct values of $\bz \in \{0,1\}^t$. 
This disjointness condition is guaranteed whenever $C$ has distance strictly larger than $t$, though this distance criteria is not a necessary condition. 
We make this condition explicit in the following definition. 
\begin{definition}
Let $C$ be a code and $E$ be a $t$-set such that any pair of codes $C_\by[E]$ and $C_\bz[E]$ are disjoint for distinct $\by,\bz \in \{0,1\}^t$.
Then we say that $C$ is $t$-disjoint with respect to $E$.
\end{definition}

Subsequent subsections obtain upper and lower bounds on the QFI of $U_\theta \, \rho_{C}[E] \, U_\theta^\dagger$ in terms of the properties of the classical code $C$.
Crucial to the development of these bounds are the \emph{weight enumerators} of $C_{{\bf z},E}$ and $C_{\bf z}[E]$, given respectively by
\begin{align}
    A_{C,k,{\bf z},E} &= |\{ \bx \in C_{{\bf z},E}  \colon \wt(\bx) = k \}|, \\
    A_{C,k,{\bf z}}[E] &= |\{ \bx \in C_{\bf z}[E] \colon \wt(\bx) = k \}|.
\end{align} 
We correspondingly define ${\mathcal X}_{C,{\bf z},E}$ and ${\mathcal X}_{C,{\bf z}}[E]$ to be random variables such that 
\begin{align}
\Pr[{\mathcal X}_{C,{\bf z},E} = k] &= \frac{ A_{C,k,{\bf z},E} }{ |C_{{\bf z},E}| }, \label{eq:Ajp-weightenum-prob}\\
\Pr[{\mathcal X}_{C,{\bf z}}[E] = k] &= \frac{ A_{C,k,{\bf z}}[E] }{ |C_{\bf z}[E]| }.\label{eq:weightenum-prob}
\end{align}
In Proposition \ref{prop:puncture-vs-nopuncture}, we prove that the variances of the random variables ${\mathcal X}_{C,{\bf z},E}$ and ${\mathcal X}_{C,{\bf z}}[E] $ are equal.
\begin{proposition}\label{prop:puncture-vs-nopuncture}
Let $C$ be any binary code of length $n$. Then, for any $E \subsetneq \{1,\dots,n\}$ and  ${\bf z} \in \{0,1\}^{|E|}$, we have $\var({\mathcal X}_{C,{\bf z},E}) = \var({\mathcal X}_{C,{\bf z}}[E] )$.
\end{proposition}
\begin{proof} 
Given any codeword $\bx \in C_{{\bf z},E}$, there is a corresponding codeword $\bx[E] \in C_{\bf z}[E]$ such that $\wt(\bx) = \wt(\bx[E]) + \wt({\bf z})$.
So, all weights of $C_{{\bf z},E}$ are a constant $\wt({\bf z})$ away from the weights of $C_{\bf z}[E]$, i.e., ${\mathcal X}_{C,{\bf z},E} = {\mathcal X}_{C,{\bf z}}[E] + \wt({\bf z})$.
Hence, this constant does not affect the variance of these associated random variables.
\end{proof}

\subsection{Lower bounds using the variances of weights}
\label{sec:2d-lower-bounds}

Given a set of erasure errors of size $t$,
our first key result in this section is a lower bound for $Q_E(C)$ in terms of the variances of the weight distributions of the codes $C_{{\bf z},E}$, for ${\bf z} \in \{ 0,1\}^t$. 
\begin{theorem} \label{thm:main-result}
Let $C$ be a binary code of length $n$, and let $E = \{j_1,\dots,j_t\} \subsetneq \{1,\ldots, n\}$ label a set of qubits that will be erased. 
Suppose that $C$ is $t$-disjoint with respect to $E$.
Then we have
\begin{align}
Q_E(C) \ge 
8 \sum_{{\bf z} \in \{0,1\}^t}   p_{{\bf z}}^2 \Var({\mathcal X}_{C,{\bf z},E} ). 
\end{align}
\end{theorem}
\begin{proof}
Recall that $\rho = \rho_C[E]$ and $H=Z_1 + \dots + Z_{n-t}$.
Using the disjointness of the codes $C_\bz$ along with the form of $\rho$ from Proposition \ref{prop:erasure-t-qubits}, we see that
\begin{align}
\rho^2
=
\frac{1}{|C|^2}
\sum_{{\bf z} \in \{0,1\}^t}
|C_{\bf z}[E]|
\sum_{\bx,\by \in C_{\bf z}[E]} |\bx\>\<\by| .\label{eq:tau2}
\end{align}

Now, for every $(n-t)$-bit string $\bx$, we have
\begin{align}
H |\bx\> &= ( -\wt(\bx) + (n-t-\wt(\bx) )|\bx\> \notag\\
&= (n-t - 2\wt(\bx))|\bx\>.
\end{align}
The cyclic property of the trace implies that $\tr(\rho^2 H^2)
    =\tr(H\rho^2 H)$. Using \eqref{eq:tau2}, and taking the trace, we get
\begin{align}
 \tr(\rho^2 H^2) 
=&
\frac{1}{|C|^2}
\!\!
\sum_{{\bf z} \in \{0,1\}^t} |C_{\bf z}[E]|
\!\!\!\!
\sum_{\bx,\by,\bu \in C_{\bf z}[E]}
\!\!
\<\bu |H|\bx\>\<\by|H|\bu\> .
\end{align}
Since $H$ is a diagonal operator in the computational basis 
it follows that 
\begin{align}
\tr(\rho^2 H^2) 
=&
\frac{1}{|C|^2}
\sum_{{\bf z} \in \{0,1\}^t}  |C_{\bf z}[E]|
\sum_{\bx  \in C_{\bf z}[E]}
\<\bx |H|\bx\>\<\bx|H|\bx\> .
\end{align}
Since $\<\bx|H|\bx\> = (n-t-2\wt(\bx))$,
we establish a connection between $\tr(\rho^2 H^2)$ and the weight enumerator through the equation
\begin{align}
\tr(\rho^2 H^2)
=
\frac{1}{|C|^2} \sum_{{\bf z} \in \{0,1\}^t}  |C_{\bf z}[E]| \, q_{\bf z},   \label{eq:piece1}
\end{align}
where
\begin{align}
 q_{\bf z} & \coloneqq 
\sum_{k=0}^{n-t} A_{C,k,{\bf z}}[E] 
\left(
(n-t)^2 - 4(n-t) k + 4k^2
\right) \label{eq:qp-defi}.
\end{align}
Similarly, we can see that 
\begin{align}
&
\tr(\rho H \rho H)
\notag\\
=&
\frac{1}{|C|^2}
\sum_{{\bf z} \in \{0,1\}^t} 
\sum_{\bx,\by,\bu,\bv \in C_{\bf z}[E]}
\<\bv |H|\bx\>\<\by|H|\bu\> \notag\\
=&
\frac{1}{|C|^2}
\sum_{{\bf z} \in \{0,1\}^t} 
\sum_{\bx,\by \in C_{\bf z}[E]}
\<\bx |H|\bx\>\<\by|H|\by\> \notag\\
=&
\frac{1}{|C|^2}
\sum_{{\bf z} \in \{0,1\}^t} 
\sum_{\bx,\by \in C_{\bf z}[E]}
(n-t-2\wt(\bx))(n-t-2\wt(\by)) \notag\\
=&
\frac{1}{|C|^2}
\sum_{{\bf z} \in \{0,1\}^t} 
\left(\sum_{\bx \in C_{\bf z}[E]}
(n-t-2\wt(\bx))\right)^2.
\end{align}
Then, by direct usage of the definitions of the classical weight enumerators, we get
\begin{align}
\tr(\rho H \rho H)
=    
\sum_{{\bf z} \in \{0,1\}^t}  \frac{r_{\bf z}^2}{|C|^2} \label{eq:piece2},
\end{align}
where
\begin{align}
 r_{\bf z} & \coloneqq \sum_{k=0}^{n-t}  A_{C,k,{\bf z}}[E] (n-t-2k) \label{eq:rp-defi}.
\end{align}
Since we have the generator bound 
$Q_E(C) \ge 
2\tr(\rho^2 H^2 ) -
2\tr(\rho H \rho H) $,
we can use \eqref{eq:piece1} and \eqref{eq:piece2} to get
\begin{align}
Q_E(C) \ge \frac{2}{|C|^2} \sum_{{\bf z} \in \{0,1\}^t}  
\left( |C_{{\bf z}}[E]| q_{\bf z} - r_{\bf z}^2 \right) . \label{eq:almost-there}
\end{align}
Using \eqref{eq:weightenum-prob} with \eqref{eq:qp-defi} and \eqref{eq:rp-defi} respectively, we get
\begin{align}
q_{\bf z} &=
 |C_{\bf z}[E]| \left(
(n-t)^2
- 4(n-t) \mathbb E({\mathcal X}_{C,{\bf z}}[E] ) 
+ 4 \mathbb E({\mathcal X}_{C,{\bf z}}[E]^2 )
\right),
\label{eq:qp-simple}\\
r_{\bf z} &= |C_{\bf z}[E]|  \left( (n-t) -  2 \mathbb E({\mathcal X}_{C,{\bf z}}[E] ) \right).  \label{eq:rp-simple}
\end{align}
Noting that 
\begin{align}
&\frac{r_{\bf z}^2}{|C_{\bf z}[E]|^2} 
=   \left( (n-t) -  2 \mathbb E({\mathcal X}_{C,{\bf z}}[E] ) \right)^2\notag\\
=&   \left( (n-t)^2 -  4(n-t) \mathbb E({\mathcal X}_{C,{\bf z}}[E] ) +  4 \mathbb E({\mathcal X}_{C,{\bf z}}[E] )^2 \right) ,\notag
\end{align}
it follows that
\begin{align}
|C_{\bf z}[E]| q_{\bf z} - r_{\bf z}^2  
&= 4 |C_{\bf z}[E]|^2 \Var({\mathcal X}_{C,{\bf z}}[E] ).\label{eq:-q-r}
\end{align}
From Proposition \ref{prop:puncture-vs-nopuncture}, it follows from \eqref{eq:-q-r}
\begin{align}
|C_{\bf z}[E]| q_{\bf z} - r_{\bf z}^2  
&= 4 |C_{\bf z}[E]|^2 \Var({\mathcal X}_{C,{\bf z},E} ).\label{eq:-q-r2}
\end{align}
Using \eqref{eq:-q-r2} on the inequality \eqref{eq:almost-there} then gives the result.
\end{proof}

{\bf Remark:} Suppose that $C$ is a linear code with distance at least 2, and has no zero columns in its generator matrix. Then for all $z=0,1$, and $j=1,\dots, n$, the cardinality of $C_{(z),\{j\}}$ is equal to $|C|/2$. 
Therefore,
\begin{align}
Q_{\{j\}}(C) \ge
2 ( \Var({\mathcal X}_{C,{(0)},\{j\}}) + \Var({\mathcal X}_{C,{(1)},\{j\}}) ).
\end{align}

Theorem~\ref{thm:main-result} shows an explicit relation between a lower bound for QFI and the classical weight enumerators of the relevant punctured codes under the $t$-qubit erasure model.
We recollect that for quantum metrology we desire the QFI to grow quadratically, or at least superlinearly, with the number of (physical) qubits $n$.
Next we show that concatenating a fixed length code $C$ with a repetition code can boost the QFI and hence potentially lead us towards our goal.

\begin{lemma}[Boosting Lemma]\label{lem:boosting}
Let $C^{\rm outer}$ be a binary code of length $m$.
Denote by $C$ the concatenated code of length $n=mr$, where $C$ is the concatenation of $C^{\rm outer}$ with a repetition code of length $r$ as the inner code, i.e., replace each bit of codewords in $C^{\rm outer}$ with $r$ copies of that bit. 
Let $E = \{i_1,\ldots, i_{t'}\} \subsetneq \{1,\ldots, n\}$ be a set that labels which $t'$ qubits are to be erased, so that the set of outer blocks with at least one erasure is given by 
$E_{g} = \{ \lfloor (i-1)/r \rfloor + 1 : i \in E \} = \{j_1,\ldots, j_{t}\} \subsetneq \{1,\ldots, m\}$, which has $t$ distinct elements.
Suppose that $C^{\rm outer}$ is $t$-disjoint with respect to $E_g$. 
Then we have
\begin{align}
Q_E(C) & \ge  
8 r^2 \sum_{{\bf z}_g \in \{0,1\}^{t}}  
p_{{\bf z}_g}^2 \Var({\mathcal X}_{C^{\rm outer},{\bf z}_g,E} ).
\end{align}
\end{lemma}
\begin{proof}
From Theorem \ref{thm:main-result} we have
\begin{align}
Q_E(C) \ge  \frac{8}{|C|^2} \sum_{{\bf z} \in \{0,1\}^{t'}}  
 |C_{{\bf z}}[E]|^2 \Var({\mathcal X}_{C,{\bf z},E} ).
\end{align}  
Due to the concatenation structure, for many values of ${\bf z}$, $ |C_{{\bf z}}[E]|$ is equal to zero. 
Hence, in the above summation, we only need to count terms where ${\bf z}$ respects the concatenation structure of $C$. 
In particular, whenever $i_k$ and $i_{k'}$ are elements of $E$ such that they belong to the same outer block, i.e.,  
\begin{align}
\lfloor (i_k - 1)/r \rfloor = \lfloor (i_{k'} - 1)/r \rfloor,
\end{align}
we must have $z_{i_k}=z_{i_{k'}}$.
Now let us label the qubits in $E$ that occur on the $j$th outer block as 
\begin{align}
S_{j} = \{i \in E:  
  \lfloor (i-1)/r \rfloor + 1 = j \}.
\end{align} 
Then, for a given ${\bf z} \in \{0,1\}^{t'}$, we know that if ``$z_{i_{k}} = z_{i_{k'}}$ for all $i_{k}, i_{k'} \in S_j$'' holds for all $j \in E_g$, then $|C_{{\bf z}}[E]| = |C^{\rm outer}_{{\bf z}_g}[E_g]|$; otherwise, $|C_{{\bf z}}[E]| = 0$.
Hence, the number of terms in the summation that are non-zero is at most $2^t$ instead of $2^{t'}$.
Next, let ${\bf x} = (x_1,\dots,x_m)$ be a codeword of $C^{\rm outer}$.
Then the corresponding concatenated codeword in $C$ is 
\begin{align}
{\bx'} = (\underbrace{x_1, \dots ,x_1}_r , \overbrace{\dots, \dots, \dots }^{(m-2)r} , \underbrace{x_m, \dots , x_m}_r ),
\end{align}
and it comprises of $m$ blocks of repeated indices. 
The weight of ${\bx'}$ is $r$ times of $\wt({\bx})$ for every $\bx \in C^{\rm outer}$.
Hence, it follows that 
\begin{align}
\Var({\mathcal X}_{C,{\bf z},E} ) = r^2 \Var({\mathcal X}_{C^{\rm outer},{\bf z}_g,S} ).
\end{align}
Since we also have $|C| = |C^{\rm outer}|$, the lemma follows.
\end{proof} 


\begin{corollary}
\label{cor:QFI_scaling}
From the above result that QFI of the concatenated code scales quadratically with the lengths of the inner repetition codes $r$, we make the following conclusions:
\begin{enumerate}

\item  When the outer code is fixed and the length of the inner repetition codes are allowed to grow, the QFI scales like $\Omega(r^2)=\Omega(n^2/m^2)$, and this is quadratic in the code length $n$ since $m$ is a constant.

\item  Since $t'$ in the boosting lemma is at most $tr$, and $C$ being $t$-disjoint fixes $t < m$, it follows that $t'=\Omega(n)$ because $r = n/m$ scales linearly in $n$, by definition, once $m$ is fixed. 
Hence, following the arguments in the previous conclusion, we can remain robust to a linear number of burst erasures while also having the QFI scale quadratically with $n$.

\end{enumerate}
\end{corollary}



So far, we have been discussing the asymptotic performance of the QFI under arbitrary erasure errors and burst erasure errors. Later in Section~\ref{sec:4a-Boosted-RM-codes}, we give more explicit lower bounds on the QFI using noisy probe states which arise from Reed-Muller codes concatenated with inner repetition codes.

\subsection{Upper bounds using the variances of weights}
 \label{sec:2e-upper-bounds}
In the previous section, we have analyzed the minimum scaling of QFI for probe states based on classical codes and also shown that concatenation with repetition codes helps achieve the Heisenberg scaling.
However, it is still an interesting problem to explore the upper limit on QFI under mild assumptions, without any concatenation.
The following result establishes such an upper bound for code-inspired probe states.

\begin{theorem} \label{thm:main-result-upperbound}
Let $C$ be a length $n$ binary code, 
and let $ E = \{j_1,\dots,j_t\} \subsetneq \{ 1,2,\ldots,n \} $.  
Suppose further that $C$ is $t$-disjoint with respect to $E$.  
Then we have
\begin{align}
&Q_E(C) \notag\\
\le&  
    16 \sum_{{\bf z} \in \{0,1\}^t} p_{\bf z} \, \mathbb E({\mathcal X}_{C,{\bf z}}[E]^2 ) 
     -
    16 \left( \sum_{{\bf z} \in \{0,1\}^t} p_{\bf z} \, \mathbb E({\mathcal X}_{C,{\bf z}}[E] ) \right)^2 .
\end{align}
\end{theorem}
\begin{proof}
To simplify notation, let $\rho = \rho_C[E]$. For every $(n-1)$-bit string $\bx$, we have
\begin{align}
H |\bx\> &= ( -\wt(\bx) + (n-1-\wt(\bx) )|\bx\> 
\notag\\
&= (n-1 - 2\wt(\bx))|\bx\>.
\end{align}
The cyclic property of the trace implies that $\tr(\rho H^2) = \tr(H\rho H)$. 
Using the form of $\rho$ from Proposition \ref{prop:erasure-t-qubits} we get
\begin{align}
 \tr(\rho H^2) 
=&
\frac{1}{|C|}
\sum_{{\bf z} \in \{0,1\}^t}
\sum_{\bx,\by,\bu \in C_{\bf z}[E]}
\<\bu |H|\bx\>\<\by|H|\bu\> .
\end{align}
Since $H$ is a diagonal operator in the computational basis,
it follows that
\begin{align}
\tr(\rho H^2) 
=&
\frac{1}{|C|}
\sum_{{\bf z} \in \{0,1\}^t} 
\sum_{\bx  \in C_{\bf z}[E]}
\<\bx |H|\bx\>\<\bx|H|\bx\> .
\end{align}
Since $\<\bx|H|\bx\> = (n-t-2\wt(\bx))$,
we establish a connection between $\tr(\rho H^2)$ and weight enumerators through the equation
\begin{align}
\tr(\rho H^2)
=
\frac{1}{|C|} \sum_{{\bf z} \in \{0,1\}^t}  q_{\bf z},   \label{eq:piece1-upper}
\end{align}
where $q_{\bf z}$ is as given in \eqref{eq:qp-defi}.
Similarly, we can see that 
\begin{align}
\tr(\rho H) =&
\frac{1}{|C|}
\sum_{{\bf z} \in \{0,1\}^t} 
\sum_{\bx \in C_{\bf z}[E]}
\<\bx |H|\bx\>.
\end{align}
It follows that
\begin{align}
\tr(\rho H) 
&=
\frac{1}{|C|}
\sum_{{\bf z} \in \{0,1\}^t} 
\sum_{\bx \in C_{\bf z}[E]}
(n-t-2\wt(\bx))
\notag\\
&=
\frac{1}{|C|}
\sum_{{\bf z} \in \{0,1\}^t} 
r_{\bf z},\label{eq:piece2-upper}
\end{align}
where $r_{\bf z}$ is as given in \eqref{eq:rp-defi}.
Since we have the generator bound 
$Q_E(C) \le 
4\tr(\rho H^2 ) - 4\tr(\rho H)^2$,
we can use \eqref{eq:piece1-upper} and \eqref{eq:piece2-upper} to get
\begin{align}
Q_E(C) 
&\le 
\frac{1}{|C|}
\left( \sum_{{\bf z} \in \{0,1\}^t}  
 q_{\bf z} 
- 
\frac{1}{|C|}
 \sum_{{\bf y},{\bf z} \in \{0,1\}^t}  
r_{\bf y} r_{\bf z} \right) .\label{eq:final-upper-bound}
\end{align}
Now note that   
\begin{widetext}
\begin{align}
\sum_{{\bf z} \in \{0,1\}^t}  q_{\bf z} 
   & = (n-t)^2 \sum_{{\bf z} \in \{0,1\}^t} |C_{\bf z}[E]| - 4 (n-t) \sum_{{\bf z} \in \{0,1\}^t} |C_{\bf z}[E]| \, \mathbb E({\mathcal X}_{C,{\bf z}}[E] ) + 4 \sum_{{\bf z} \in \{0,1\}^t} |C_{\bf z}[E]| \, \mathbb E({\mathcal X}_{C,{\bf z}}[E]^2 ) \\
   & = (n-t)^2 |C| - 4 (n-t) \sum_{{\bf z} \in \{0,1\}^t} |C_{\bf z}[E]| \, \mathbb E({\mathcal X}_{C,{\bf z}}[E] ) + 4 \sum_{{\bf z} \in \{0,1\}^t} |C_{\bf z}[E]| \, \mathbb E({\mathcal X}_{C,{\bf z}}[E]^2 ),\label{eq:38}
\end{align}
and
\begin{align}
\frac{1}{|C|} \left( \sum_{{\bf z} \in \{0,1\}^t} r_{\bf z} \right)^2 
   & = \frac{1}{|C|} \left( (n-t) \sum_{{\bf z} \in \{0,1\}^t} |C_{\bf z}[E]| - 2 \sum_{{\bf z} \in \{0,1\}^t} |C_{\bf z}[E]| \, \mathbb E({\mathcal X}_{C,{\bf z}}[E] )  \right)^2 \\
   & = \frac{1}{|C|} \left( (n-t) |C| - 2 \sum_{{\bf z} \in \{0,1\}^t} |C_{\bf z}[E]| \, \mathbb E({\mathcal X}_{C,{\bf z}}[E] )  \right)^2 \\
   & = (n-t)^2 |C| - 4 (n-t) \sum_{{\bf z} \in \{0,1\}^t} |C_{\bf z}[E]| \, \mathbb E({\mathcal X}_{C,{\bf z}}[E] ) + \frac{4}{|C|} \left( \sum_{{\bf z} \in \{0,1\}^t} |C_{\bf z}[E]| \, \mathbb E({\mathcal X}_{C,{\bf z}}[E] ) \right)^2 . 
   \label{eq:tolowerbound} 
\end{align}
\end{widetext}
The result then follows.
\end{proof} 

\begin{remark}
\label{rem:simple-upper-bound}
We can obtain a simpler, albeit looser, upper bound on the QFI that is expressed explicitly in terms of the variances of the weight distributions of the shortened codes. To arrive at a simpler bound, aside from the assumptions made in Theorem \ref{thm:main-result-upperbound}, we suppose that we additionally have 
\begin{align}
 \sum_{{\bf z} \in \{0,1\}^t} p_{\bf z}\,  \mathbb E({\mathcal X}_{C,{\bf z}}[E]) \ge s,
\end{align}
for some $s \ge 1$. 
Then 
\begin{align}
Q_E(C) \le& 16 \left( \frac{s}{n} \right) \sum_{{\bf z} \in \{0,1\}^t} p_{\bf z}\, \var({\mathcal X}_{C,{\bf z}}[E] )
\notag\\
&+
16 \left( 1- \frac{s}{n} \right) \sum_{{\bf z} \in \{0,1\}^t} p_{\bf z}\, \mathbb E({\mathcal X}_{C,{\bf z}}[E]^2 ).\label{eq:simple-upper-bound}
\end{align}
\end{remark}
\begin{proof}[Proof of \eqref{eq:simple-upper-bound} in Remark \ref{rem:simple-upper-bound}]
Now let us consider an inequality relating 
$\sum_{\bz}    p_\bz x_\bz^2$ and 
$n (\sum_{\bz} p_\bz x_\bz)^2$, where $p_\bz$ are probabilities, and $x_\bz$ are non-negative numbers in the interval $[0,n]$ for some positive number $n$. Suppose further that $\sum_{\bz} p_\bz x_\bz \ge s \ge 1$.
Then it follows that
\begin{align}
\sum_{\bz} p_\bz x_\bz^2 
\le 
n\sum_{\bz} p_\bz x_\bz 
\le 
\left( \frac{n}{s} \right) \left(\sum_{\bz} p_\bz x_\bz \right)^2 . \label{eq:special-inequality}
\end{align}
Using \eqref{eq:special-inequality} and identifying $x_\bz$ with $\mathbb E({\mathcal X}_{C,{\bf z}}[E])$, we find that $
    - \left(\sum_{\bz} p_\bz x_\bz \right)^2
    \le
    \left( \frac{s}{n} \right)
    \sum_{\bz} p_\bz x_\bz^2 .$
Substituting this inequality into Theorem \ref{thm:main-result-upperbound} gives the upper bound \eqref{eq:simple-upper-bound}.
\end{proof} 

Theorem \ref{thm:main-result-upperbound} shows that the QFI upper bound depends on a variance-like quantity on the weight distributions of the $2^t$ shortened codes. 
If $t \le k$ and the submatrix comprising the columns corresponding to $E$ of any generator matrix for $C$ has full rank, then we will indeed have $p_{\bf z} = \frac{1}{2^t}$ by symmetry of binary subspaces.
Furthermore, both the lower and upper bounds on the QFI indicate that we need codes with a large variation in codeword weights.
So, for such codes it is reasonable to expect that $\mathbb E({\mathcal X}_{C,{\bf z}}[E]) \approx \frac{n}{2}$ whenever $t \ll n$, in which case the QFI lower and upper bounds simplify to 
\begin{align} 
\frac{4}{2^{t+1}} \left[ \frac{4}{2^t} \sum_{{\bf z} \in \{0,1\}^t} \mathbb E({\mathcal X}_{C,{\bf z}}[E]^2 ) - n^2 \right] 
\le Q_E(C) 
\notag\\
Q_E(C) 
\le  
4
\left[
    \frac{4}{2^t} \sum_{{\bf z} \in \{0,1\}^t}   \, \mathbb E({\mathcal X}_{C,{\bf z}}[E]^2 ) - n^2
    \right]. \label{II.82}
\end{align} 
We like to emphasize that \eqref{II.82} holds for any positive integer $n$.

Our comparison shows quite clearly that the generator-based QFI lower and upper bounds from the literature are away by a factor $2^{-t-1}$ in our setting of code-inspired probe states.
But, it also explicitly shows that codes with quadratically scaling second moments on the weight distributions of their shortened versions are highly desirable for robust metrology.

\section{An explicit observable} \label{sec:3-measurements}

For us, $\rho = \rho_E(C)$ is not a pure state, but a mixed state with the form 
\begin{align}
\rho =   \sum_{\bz} \pz   |\psiz\>\<\psiz|  .
\end{align}
Unfortunately, the symmetric logarithmic derivative (of $\rho_\theta = U_\theta \rho U_\theta^\dagger$) becomes more complicated in this case. 
We can nonetheless consider the operator 
 \begin{align}
L = i \sum_{\bz} \pz ( |\psiz\>\<\psiz| H - H |\psiz\>\<\psiz|  )
\label{eq:approximate-SLD}
\end{align}
as the observable to measure on
\begin{align}
\rho_\theta = U_\theta \rho U_\theta^\dagger = \sum_{\bz} \pz |\psiz^\theta\>\<\psiz^\theta|,
\label{eq:final-encoded-probe-state}
\end{align}
where $H = Z_1 + \dots + Z_{n-t}$, $|\psiz^\theta\> \coloneqq U_\theta |\psiz\>$, and $U_\theta = e^{-i \theta H}$. 
To evaluate the performance of our observable $L$, we must evaluate the quantities
\begin{align}
\tr(\rho_\theta L ), \quad \frac{\partial}{\partial \theta}\tr(\rho_\theta L )  , \quad \tr(\rho_\theta L^2 ).
\end{align}
These quantities can be calculated using the following lemma.
\begin{lemma}
\label{lem:some-inner-products}
For every $\by \in \{0,1\}^t$, let 
$|\psiy\> = \frac{1}{\sqrt{|C_\by|}} \sum_{\bx \in C_\by}   |\bx\>$, 
where $C_{\by}$ is a binary code of length $n-t$. 
Then 
\begin{widetext}
\begin{align}
\<\psiy^\theta |\psiy \> &= \frac{1}{|C_\by|} e^{i {\theta} (n-t)} \sum_{\bx\in C_\by}
e^{-2i{\theta} \wt(\bx)} , \\
\<\psiy^\theta |H|\psiy \> &= \frac{1}{|C_\by|}
e^{i {\theta} (n-t)}  \sum_{\bx\in C_\by} 
((n-t) - 2 \wt(x)) 
e^{-2i{\theta} \wt(\bx)} ,\\
\<\psiy |H|\psiy \> &=    (n-t)- \frac{2}{|C_\by|} \sum_{\bx\in C_\by}\wt(\bx) , \\
\<\psiy |H^2|\psiy \> &=   (n-t)^2- \frac{4(n-t)}{|C_\by|} \sum_{\bx\in C_\by}\wt(\bx) 
+ \frac{4}{|C_\by|} \sum_{\bx\in C_\by}\wt(\bx)^2 .
\end{align}
\end{widetext}
\end{lemma}
\begin{proof}
It is easy to see that 
\begin{align}
H|\bx\> &= (n-t)-2\wt(\bx)|\bx\>,\\
H^2|\bx\> &= ((n-t)-2\wt(\bx))^2|\bx\>,\\
e^{-i\theta H}|\bx\> &= \exp[-i {\theta} ((n-t)-2\wt(\bx))]|\bx\>.
\end{align}
The results then directly follow from these observations.  
\end{proof} 
Now define
\begin{align}
\mu_\by & \coloneqq \frac{1}{|C_\by|} \sum_{\bx\in C_\by} \wt(\bx), \\
V_\by & \coloneqq \frac{1}{|C_\by|} \sum_{\bv \in C_\by} \left( \wt(\bv)^2 \right) - \mu_\by^2.
\end{align}
Here, $V_\by$ denotes the variance of the weight distribution of $C_\by$. 
Then we have the following lemmas.
\begin{lemma}
The unitarily evolved probe state and its associated symmetric logarithmic derivative satisfy
\begin{align}
\tr(\rho_\theta L)  &=
8\theta \sum_{\by} \py^2 V_\by +O(\theta^3) \label{eq:bias}
\\
\frac{\partial \rho_\theta}{\partial \theta}\tr(\rho_\theta L)  &= 
8 \sum_{\by} \py^2 V_\by  +O(\theta^2) .
 \end{align}
\end{lemma}
\begin{proof}
Since $H$ is a diagonal matrix and we expand $|\psiy\>$ in the computational basis, we have 
\begin{align}
&\tr(\rho_\theta L) 
\notag\\
=&
2i \sum_{\bu,\by} \pu \py \tr\bigg[ 
	|\psiu^\theta\>\<\psiu^\theta| \
	\bigg( |\psiy\>\<\psiy|H-H|\psiy\>\<\psiy| \bigg)
\bigg]\notag\\
=&
2i \sum_{\by} \py^2  \bigg[ 
	\<\psiy|H|\psiy^\theta\>\<\psiy^\theta|   \psiy\>
-		\<\psiy| \psiy^\theta\>\<\psiy^\theta|H|\psiy\>
\bigg].
\end{align}
Now by interpreting $ \<\psiy|H|\psiy^\theta\> \<\psiy^\theta|\psiy\> $ as a complex number $z$, and noting that $z-z^*=2i{\rm Im}(z)$, we can use the results of Lemma \ref{lem:some-inner-products} to get
\begin{widetext}
\begin{align}
\tr(\rho_\theta L) 
&=
-4 \sum_{\by} \frac{\py^2}{|C_\by|^2} {\rm Im}\left[ 
  \sum_{\bx\in C_\by}e^{-2i{\theta} \wt(\bx)} 
  \sum_{\bv\in C_\by}  
((n-t) - 2 \wt(\bv)) 
e^{2i{\theta} \wt(\bv)}  
\right] \notag\\
&=
-4 \sum_{\by} \frac{\py^2}{ |C_\by| ^2 }  
  \sum_{\bx,\bv\in C_\by}
  (  n-t - 2 \wt(\bv)  ) 
 \bigg[
  \cos(2{\theta} \wt(\bx))  
 \sin(2{\theta} \wt(\bv))  
 -
 \sin (2{\theta}\wt(\bx) )  
\cos(2{\theta} \wt(\bv))  
\bigg]   \notag\\
&=
-4 \sum_{\by} \frac{\py^2}{ |C_\by| ^2 }  
  \sum_{\bx,\bv\in C_\by}
  (  n-t - 2 \wt(\bv)  ) 
 \sin(2{\theta}(\wt(\bv) - \wt(\bx)  )    .
\end{align}
It follows that
\begin{align} 
\tr(\rho_\theta L) 
&= 0
-8\theta \sum_{\by} \frac{\py^2}{ |C_\by| ^2 }  
  \sum_{\bx,\bv\in C_\by}
  (  n-t - 2 \wt(\bv)  ) 
(\wt(\bv) - \wt(\bx)  )    + O(\theta^3)
\notag\\
&=
16\theta \sum_{\by} \frac{\py^2}{ |C_\by| ^2 }  
  \sum_{\bx,\bv\in C_\by}
  (   \wt(\bv)^2 - \wt(\bv)\wt(\bx)  )    + O(\theta^3)\notag\\
&=
 16\theta \sum_{\by} \py^2  V_\by    + O(\theta^3) 
\end{align}
\end{widetext}
Differentiating with respect to $\theta$, we get 
\begin{align}
\frac{\partial}{\partial \theta}
\tr(\rho_\theta L) 
 &=
 8 \sum_{\by} \py^2 V_\by  +O(\theta^2) .  \hfill \qedhere
\end{align}
\end{proof}
Similarly we prove another property. 
\begin{lemma}
The unitarily evolved probe state and its associated symmetric logarithmic derivative satisfy
\begin{align} 
\tr(\rho_\theta L^2) 
	&= 
16 \sum_{ \by }  \py^3 V_\by + O(\theta^2),
\end{align}\end{lemma}
\begin{proof}
Next, using the assumption that $|\psiy\>$ and $|\psiz\>$ are term-wise orthogonal since $C_\by \cap C_\bz = \emptyset$, we find that 
\begin{widetext}
\begin{align}
\tr(\rho_\theta L^2)
&= 
-4 \sum_{\bu,\by,\bz} \pu \py \pz \tr\bigg[ 
	|\psiu^\theta\>\<\psiu^\theta| \ 
	\bigg( |\psiy\>\<\psiy|H-H|\psiy\>\<\psiy| \bigg) \ 
	\bigg( |\psiz\>\<\psiz|H-H|\psiz\>\<\psiz| \bigg)
\bigg]
\notag\\
&= 
-4 \sum_{ \by }  \py^3 \tr\bigg[ 
	|\psiy^\theta\>\<\psiy^\theta| \ 
	\bigg( |\psiy\>\<\psiy|H-H|\psiy\>\<\psiy| \bigg) \ 
	\bigg( |\psiy\>\<\psiy|H-H|\psiy\>\<\psiy| \bigg)
\bigg]
\notag\\
&= 
-4 \sum_{ \by }  \py^3 
\<\psiy^\theta| \ \bigg[ 
	\bigg( |\psiy\>\<\psiy| H \bigg) \ \bigg( |\psiy\>\<\psiy| H \bigg)
	- \bigg( |\psiy\>\<\psiy| H \bigg) \ \bigg( H |\psiy\>\<\psiy| \bigg)
	\bigg.
	\notag\\
	&\qquad\qquad\qquad\qquad \bigg.
	- \bigg( H |\psiy\>\<\psiy| \bigg) \ \bigg( |\psiy\>\<\psiy| H \bigg)
	+ \bigg( H |\psiy\>\<\psiy| \bigg) \ \bigg( H |\psiy\>\<\psiy| \bigg)
\bigg]
	|\psiy^\theta\>
	\notag\\
	&= 
-4 \sum_{ \by }  \py^3 
 \bigg[ 
\<\psiy^\theta| \psiy\>\<\psiy|H|\psiy\>\<\psiy|H|\psiy^\theta\>
	-\<\psiy^\theta |\psiy\>\<\psiy|H^2|\psiy\>\<\psiy| \psiy^\theta\>
	\bigg.
	\notag\\
	&\qquad\qquad\qquad \bigg.
	-\<\psiy^\theta|H|\psiy\> \<\psiy|H|\psiy^\theta\>
	+\<\psiy^\theta|H|\psiy\>\<\psiy|H|\psiy\>\<\psiy|\psiy^\theta\>
\bigg]
\notag\\
	&= 
-4 \sum_{ \by }  \py^3 
 \bigg[ 
2{\rm Re}\bigg( \<\psiy^\theta| \psiy\>\<\psiy|H|\psiy\>\<\psiy|H|\psiy^\theta\> \bigg)
	-|\<\psiy^\theta |\psiy\>|^2\<\psiy|H^2|\psiy\> 
	-|\<\psiy^\theta|H|\psiy\>|^2
\bigg].
\end{align}
\end{widetext}
Then it follows from Lemma \ref{lem:some-inner-products} that
\begin{widetext}
\begin{align}
&{\rm Re}\bigg( \<\psiy^\theta| \psiy\>\<\psiy|H|\psiy\>\<\psiy|H|\psiy^\theta\> \bigg)
\notag\\
&=
\frac{1}{|C_\by|^2}
{\rm Re}\left(
\sum_{\bu\in C_\by}e^{-2i{\theta}\wt(\bu)}  
(  n-t- 2 \mu_\by)
\sum_{\bx\in C_\by} 
(n-t - 2 \wt(\bx)) e^{2i{\theta} \wt(\bx)}  
\right)
\notag\\
&=
\frac{
(  n-t- 2 \mu_\by)}{|C_\by|^2}
\sum_{\bu, \bx\in C_\by} 
(n-t - 2 \wt(\bx)) 
\bigg[
\cos(2{\theta}\wt(\bu))  
\cos(2{\theta} \wt(\bx) )   
+
\sin(2{\theta} \wt(\bu))
\sin(2{\theta} \wt(\bx))
\bigg]
\notag\\
&=
\frac{
(  n-t- 2 \mu_\by)}{|C_\by|^2}
\sum_{\bu, \bx\in C_\by} 
(n-t - 2 \wt(\bx)) 
\cos(2{\theta} (\wt(\bu) - \wt(\bx) ) .
\end{align}
Thus,
\begin{align}
&
{\rm Re}\bigg( \<\psiy^\theta| \psiy\>\<\psiy|H|\psiy\>\<\psiy|H|\psiy^\theta\> \bigg)
=
(  n-t- 2 \mu_\by)^2 +O(\theta^2).
\end{align}
Next we find that
\begin{align}
|\<\psiy^\theta |\psiy\>|^2\<\psiy|H^2|\psiy\> 
&=
\frac{1}{ |C_\by| ^2} \sum_{\bv,\bu \in C_\by} e^{2i{\theta} \wt(\bv)} e^{-2i {\theta} \wt(\bu)} 
 \left(
   (n-t)^2- 4(n-t)\mu_\by
+ \frac{4}{|C_\by|} \sum_{\bx\in C_\by}\wt(\bx)^2
 \right).
\end{align}
Hence, it follows that
\begin{align}
|\<\psiy^\theta |\psiy\>|^2\<\psiy|H^2|\psiy\> 
&=   
   (n-t)^2- 4(n-t)\mu_\by
+ \frac{4}{|C_\by|} \sum_{\bx\in C_\by}\wt(\bx)^2
+O(\theta^2)
\notag\\
&=   
  (n-t - 2\mu_\by)^2 + 4 V_\by +O(\theta^2).
\end{align}
Also, 
\begin{align}
&|\<\psiy^\theta|H|\psiy\>|^2 =
\left| \frac{1}{|C_\by|} \sum_{\bx\in C_\by} 
((n-t) - 2 \wt(\bx)) \ e^{2i{\theta} \wt(\bx)} \right|^2,
\end{align}
and hence
\begin{align}
|\<\psiy^\theta|H|\psiy\>|^2 
&= |n-t-2\mu_\by  +i(n-t)\theta \mu_\by - \frac{2 i\theta}{|C_\by|} \sum_{\bx \in C_\by}\wt(\bx)^2 |^2
+O(\theta^2)
\notag\\
&=(n-t-2\mu_\by)^2 + O(\theta^2).
\end{align}
\end{widetext}
The result then follows.
\end{proof}
Now we present the main result of this section.
\begin{theorem}\label{thm:SLD-measurement}
Let $\rho = \sum_{\by \in \{0,1\}^t} p_\by |\psiy\>\<\psiy|$, where $|\psiy\>$ satisfy the assumptions of Lemma \ref{lem:some-inner-products}. 
Assume that the length $(n-t)$ (potentially non-linear) codes $C_\by$ satisfy $C_\by \cap C_\bz = \emptyset$ for distinct $\by $ and $\bz$, so that $|\psiy\>$ and $|\psiz\>$ are pairwise orthogonal, term by term. 
Let $\rho_\theta$ and $L$ be as defined in \eqref{eq:final-encoded-probe-state} and \eqref{eq:approximate-SLD} respectively. 
When $\theta$ is small, measuring the observable $L$ gives an estimator $\hat \theta$ of $\theta$ that satisfies the bound
\begin{align}
 \Var(\hat \theta)
&\le
\frac{1 } 
{16 \sum_{\by \in \{0,1\}^t} \py^2 V_\by } + O(\theta^2).
\end{align}
Moreover, when $\theta =0$, $\hat \theta$ is also a locally unbiased estimator.
\end{theorem}
\begin{proof}
Using the error propagation formula and the previous lemmas, we have that
\begin{align}
\Var(\hat \theta)
=&
\frac{\tr(\rho_{\theta} {L}^2)- \tr(\rho_{\theta} L)^2 }
    {\left|\frac{\partial}{\partial \theta}   \tr(\rho_{\theta} L) \right|^2} +O(\theta^2)  \notag\\
&=\frac{ 16  \sum_{ \by }  \py^3 V_\by  +O(\theta^2)} 
{16^2 \left(\sum_{\by} \py^2 V_\by +O(\theta^2) \right)^2}+O(\theta^2)
\notag\\
&\le
\frac{1 } 
{16   \sum_{\by} \py^2 V_\by } +O(\theta^2).   
\end{align}
The fact that $\hat \theta$ is approximately a locally unbiased estimator at $\theta=0$ arises from \eqref{eq:bias}.
\end{proof} 
Let us now identify $C_\by$ with $C_{\by}[E]$ for each $\by \in \{0,1\}^t$ and thus $\rho$ with $\rho_C[E]$, the classical code-inspired probe state after the subset $E$ of qubits is erased.
Then, we can compare the result of Theorem \ref{thm:SLD-measurement} on $\Var(\hat \theta)$ with respect to the observable $L$ with the lower bound on QFI obtained in Theorem~\ref{thm:main-result}.
Recollecting that the minimum variance is given by the inverse of QFI, we see that measuring the observable $L$ is optimal asymptotically, i.e., whenever $\theta$ is close to zero.
Furthermore, when combined with the boosting lemma, this implies that measuring $L$ on a probe state constructed from a fixed length ($m$) code concatenated with a length $\Omega(n)$ (inner) repetition code will have the desired Heisenberg scaling in $\Var(\hat \theta)$ whenever $\theta$ is close to zero. 
Finally, we note that the above result holds also when $\theta$ approaches integer multiples of $\pi$.

\section{Examples}\label{sec:4-examples}

\subsection{Boosted Reed-Muller codes} \label{sec:4a-Boosted-RM-codes}

Let us revisit the $[[8,3,2]]$ quantum Reed-Muller code based probe state discussed in Section~\ref{sec:2b-RM-codes}.
In this case the classical code corresponding to the probe state is the $[8,4,4]$ self-dual Reed-Muller code $C = \text{RM}(1,3)$.
We observed earlier that under no erasure the state has a QFI lower bound of $16$ but under a single erasure this drops to $7$.
Assume we concatenate $C^{\rm outer} = C$ with an inner repetition code of length $r = 3$ to get a code of length $n = 24$.
According to the boosting lemma (Lemma \ref{lem:boosting}), the QFI lower bound for the probe state constructed from the concatenated code only depends on the projection of the erasures to the outer code $C$.
Since we earlier considered a single erasure on the non-concatenated code, in order to make a fair comparison let us fix the erasure rate as $1/8$.
Thus, approximately three qubits get erased on the $24$-qubit probe state.
If these qubit indices belong to the same ``block'' of repeated bits in the concatenated code, then the projection to $C^{\rm outer} = C$ produces a single qubit erasure.
While this produced a QFI lower bound of $7$ for the non-concatenated code, for the $24$-qubit probe state this is enhanced by $r^2 = 9$ to $63$, according to the boosting lemma. 
Similarly, if the projection produces two (resp. three) erasures on the outer code, then the QFI for the $24$-qubit probe state is $27$ (resp. $9$).
In general, for the outer code $C$, if one, two or three qubits are erased, then the normalized QFI lower bound is $7/8,3/8$, and $1/8$, respectively.
If four or more qubits are erased then the QFI lower bound is trivial (i.e., $0$).
Therefore, when concatenated with a repetition code of length $r$, the normalized bound increases to $7r/8, 3r/8$, and $r/8$, respectively, depending on the size of the projection of the erasures on the concatenated code to just the outer code. 

In Figure \ref{fig:comparisons} we compare the performance of our probe state from the concatenated RM(1,3) code with that of previously studied GNU probe states \cite{ouyang2019robust}. GNU probe states arise from the codespace of a specific family of permutation-invariant quantum error correction codes called GNU codes \cite{ouyang2014permutation}. These codes on $GNU$ qubits have three parameters, given by $G$, $N$ and $U$. Here, $G$ relates to the correctible number of bit-flips, $N$ corresponds to the number of correctible phase-flips, and $U$ is an unimportant scaling factor that is at least 1. These permutation-invariant codes however cannot be studied using the framework in our paper, as the distance of the corresponding classical code is equal to 1.

\begin{figure}[!htb] 
\centering
\includegraphics[scale=0.4]{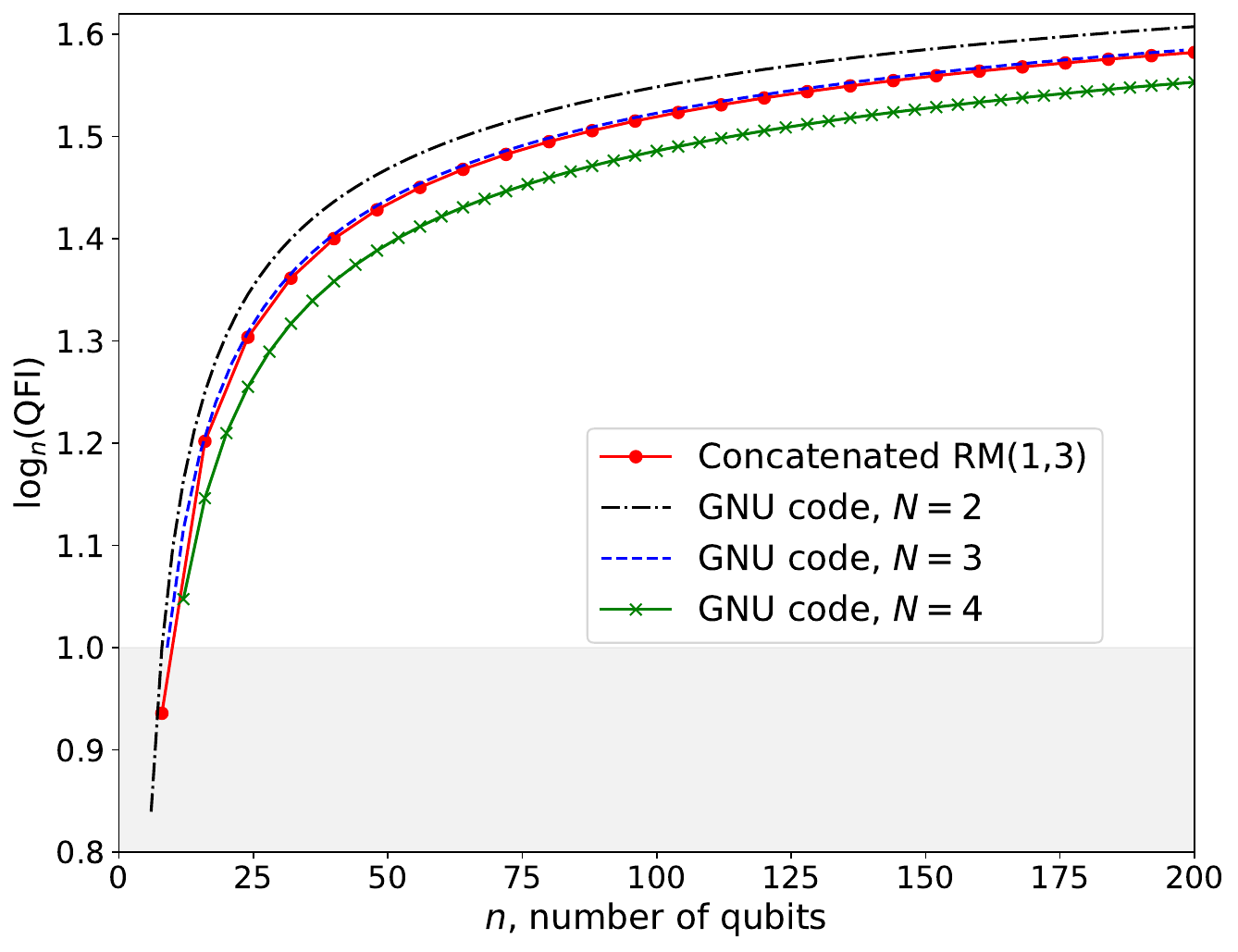}
\caption{We plot lower bounds on $\log_n(\QFI)$ 
for various code-inspired probe states for the robust quantum metrology problem. We compare the lower bounds that we have for the concatenated RM(1,3) code with that of previously studied GNU probe states \cite{ouyang2019robust} after one erasure error has occurred. Whenever the lower bound is above 1, there is a quantum advantage in using these code-inspired probe states.  
}
\label{fig:comparisons}
\end{figure}

\subsection{Boosted CSS codes}
\label{sec:4b-boosted-css}
 
As mentioned earlier, the general code-inspired probe state we have considered is always the logical $\ket{++\cdots +}$ state of a CSS code whose logical $X$ group (including the $X$-type stabilizers) is given by the chosen classical binary code $C$, as long as $C$ is a linear code.
So, if we used $C = \text{RM}(1,3)$ above, then in the future when QEC-based metrology becomes feasible, we will only be able to detect a single error since the corresponding CSS code has parameters $[[8,3,2]]$.
However, if we chose the logical $\ket{+}$ state of the $[[15,1,3]]$ quantum Reed-Muller code, then we can make use of Reed-Muller properties while also being able to correct a single error.
Some properties that could be leveraged are the large symmetry group of (classical) Reed-Muller codes~\cite{Macwilliams-1977} and the fact that this quantum code has a transversal $T$ property~\cite{Bravyi-pra12,Rengaswamy-jsait20}.
Since transversal $T$ realizes logical $T^{-1}$ on this $[[15,1,3]]$ code, it does not take the logical $\ket{+}$ state to a code state that is orthogonal to it, so it remains unclear how this symmetry can be leveraged.
However, since the unitary induced by the generator $H = \sum_i Z_i$ produces a transversal $Z$-rotation, it will be interesting to explore the utility of the transversal $Z$-rotation property.
If this is found to be useful for quantum metrology, then one can easily incorporate well-known families of CSS codes, such as triorthogonal codes~\cite{Bravyi-pra12}, that possess such a property into our code-inspired probe state framework~\cite{Rengaswamy-jsait20,Rengaswamy-isit20}. 

Surface codes form a popular family of CSS codes that are thought to be attractive candidates for quantum error correction in the near-term~\cite{Fowler-arxiv12}.
Although these codes encode a fixed number of qubits, with typical parameters being $[[2d^2,2,d]]$ on a $d \times d$ square lattice, for metrology purposes our results show that only the variances of the weight distributions of the corresponding shortened classical codes matter.
It is known that surface codes can be constructed as a hypergraph product of two classical length $d$ repetition codes~\cite{Krishna-arxiv19a,Krishna-arxiv19b}.
Since repetition codes only have codeword lengths $0$ and $n$, they have a quadratically scaling variance even under erasures.
However, the logical $X$ group for surface codes is not given by a repetition code, so one needs to analyze the weight distribution of this group to assess the utility of the resultant probe state for robust metrology.
As surface codes are highly likely to be practically realized, this approach would naturally be adaptable to fault-tolerant quantum error correction based metrology when that becomes feasible. 
 
 Our scheme has some interesting connections with \cite{dur2014improved}, where a scheme for quantum metrology with active quantum error correction was proposed. There, the probe states were of the form $(|0_L\>^{\otimes m} + |1_L\>^{\otimes m} )/\sqrt 2$, where $|0_L\>$ and $|1_L\>$ are logical codewords from any quantum error correction code. So the concatenation has the repetition code as the outer code, and other quantum error correction codes as inner codes, opposite to the case we considered. 
 
 Another related work is in \cite{lu2015robust}, where the authors derive some conditions for the noise model under which the QFI has absolutely no degradation.
 In contrast, we consider a weaker condition, where the QFI can degrade under the effects of erasure errors. In \cite{lu2015robust}, the authors revisited the metrology problem using the probe states $(|0_L\>^{\otimes m} + |1_L\>^{\otimes m} )/\sqrt 2$, and obtained heuristically the same conclusion as we do.
 Namely, they also find that concatenation of quantum error correction codes with repetition codes is advantageous. 
 In our work however, we have several additional key findings. 
 First, we have explicit bounds for the QFI in this scenario that are absent in \cite{lu2015robust}, which applies to any quantum error correction code concatenated with repetition codes.
 Second, to the best of our knowledge, our work is the first to establish the connections of the problem of robust quantum metrology with that of coding theory.
 

\section{Discussions}
 \label{sec:5-discussions}
 
 In summary, we have studied the performance of code-inspired probe states for the problem of noisy quantum metrology. We find that there is a strong connection between noisy quantum metrology and classical coding theory. Namely, the QFI is related to the variances of the weight distributions of shortened codes. The larger the variance, the larger the corresponding QFI.
 Moreover, we have a boosting lemma that implies that any CSS code, when concatenated with repetition codes of linear length can be useful for robust field-sensing with a constant number of erasure errors%
\footnote{These boosted probe states have a similar form to those proposed in \cite{dur2014improved}, and indeed, our results are reminiscent of those in \cite{lu2015robust}.}. 
 We thereby side-step the no-go result of random codes for robust field-sensing by having these CSS codes to have asymptotically vanishing relative distance.
 We also expect that when the CSS codes are concatenated with repetition codes, we will also do very well for burst erasure errors, but we leave this for future work.

In general, erasure errors do not commute with the signal, and therefore their impact in the different stages of the quantum sensing protocol is different. In our model, we assume that only errors occur during signal accumulation, which models the scenario where the dominant noise process occurs before signal accumulation. When erasures occur during state preparation, one would expect the QFI to degrade more than if the erasures occur later. This is because if $t$ erasure errors occur at the end of signal accumulation just before measurement, for QEC codes that correct at least $t$ errors, then the erasure errors do not decrease the QFI \cite{huang2022imaging}. More recently, the active QEC has been shown to be an effective way to combat erasure errors (and deletion errors \cite{ouyang2021permutation}) that occur during signal accumulation \cite{ouyang2022quantum}.

We like to highlight the distinction between our protocol and the usual QEC setting. In the usual QEC setting, the goal is to minimize the logical error rate using active QEC given some predetermined amount of noise. In our setting, we want to choose the best QEC states to maximize the QFI, without the use of active QEC. In particular, our setting does not require the logical error rate to be low; only the QFI is the metric of merit here. Hence, we like to emphasize that while our probe state is indeed a logical state of some CSS code, the advantage we get after erasures is only in terms of precision in the context of metrology and not in the logical error rate for the code.

 There are many open problems that remain to be solved. 
First, continuous quantum error correction protocols have previously been studied \cite{sarovar2005continuous,muller2014quantum,reiter2017dissipative}.
 It will be interesting to extend our work further in this direction, to see how  continuous time quantum error correction can be integrated with robust quantum metrology.  
Second, the potential of using quantum Reed-Muller codes for fault tolerant quantum metrology has recently been investigated \cite{Kapourniotis2019}. 
Since quantum Reed-Muller codes are CSS codes, it is interesting to see how quantum Reed-Muller codes concatenated with repetition codes would perform correspondingly in a fault-tolerant setting for quantum metrology. 
Third, it will also be interesting to see how concatenation of random codes with specific families of codes with structure will perform for robust quantum metrology, as this will correspondingly extend the work of \cite{PRX-random-states} which studied noisy quantum metrology for fully random quantum states.
Fourth, it will be interesting to extend our results to a multiparameter setting, using recent developments on obtaining tight bounds for the robust estimation of incompatible observables \cite{sidhu2021tight,hayashi2022tight}.

 
\section*{Acknowledgements} \label{sec:ack}

YO acknowledges support from EPSRC (Grant No. EP/M024261/1)
and the QCDA project (Grant No. EP/R043825/1) which has received funding from the QuantERA ERANET Cofund in Quantum Technologies implemented within the European Union’s Horizon 2020 Programme.
YO is supported in part by NUS startup
grants (R-263-000-E32-133 and R-263-000-E32-731), and the
National Research Foundation, Prime Minister’s Office, Singapore and the Ministry of Education, Singapore under the Research Centres of Excellence programme.
The work of NR was supported in part by the National Science Foundation (NSF) under Grant Nos. 1718494 and 1908730.

\bibliography{robustmetro}{}
\bibliographystyle{IEEEtran}

\appendices

\section{Background on Calderbank-Shor-Steane (CSS) Codes}
\label{sec:css_codes}

An $[n,k,d]$ classical binary linear code $C$ is a $k$-dimensional subspace of $\mathbb{F}_2^n$, the vector space of all length-$n$ binary vectors. 
It encodes $k$ message bits, $m$, into a length-$n$ codeword, $c$, through a $k \times n$ generator matrix, $G(C)$, as $c = m G(C)$.
The code has minimum distance $d$, which means that the Hamming weight (i.e., number of non-zero entries) of any codeword is $d$.
The dual code to $C$, denoted $C^\perp$, is the subspace orthogonal to $C$ in $\mathbb{F}_2^n$.

The CSS construction takes as input an $[n,k_1,d_1]$ code $C_1$ and an $[n,k_2,d_2]$ code $C_2$ such that $C_2 \subseteq C_1$, and produces an $[[n, k = k_1-k_2, d \ge \min\{ d_1, d_2^\perp \}]]$ quantum stabilizer code, where $d_2^\perp$ denotes the minimum distance of $C_2^\perp$.
Such a code is said to encode $k$ \emph{logical} qubits into $n$ \emph{physical} qubits.
Each codeword in $C_2$ produces an $X$-stabilizer by mapping $1$s to Pauli $X$s and $0$s to $I$s (identity).
Similarly, each codeword in $C_1^\perp$ produces a $Z$-stabilizer by mapping $1$s to Pauli $Z$s and $0$s to $I$s.
The encoding map for the CSS code is defined as follows.
Given a binary vector $x \in \mathbb{F}_2^k$, which represents the logical basis state $|x\>_L$, the encoded state is given by
\begin{align}
|x\>_L \longmapsto |\psi_x\> \coloneqq \frac{1}{\sqrt{|C_2|}} \sum_{c \in C_2} |x G(C_1/C_2) \oplus c \>,
\end{align}
where $G(C_1/C_2)$ denotes a generator matrix for the quotient space $C_1/C_2$ and $\oplus$ represents modulo $2$ addition of binary vectors.
It can be easily verified that any $X$- or $Z$-stabilizer defined above preserves this state, as required.

\end{document}